\documentclass[sigconf, 10pt,nonacm]{acmart}
\usepackage{amsmath}
\usepackage{amsthm}
\usepackage{graphicx}
\usepackage{xspace}
\usepackage[inline]{enumitem}
\usepackage{hhline}
\usepackage{subfigure}
\usepackage{multirow}
\usepackage{colortbl}
\usepackage{balance}
\usepackage{textcomp}
\usepackage[ruled,vlined,linesnumbered]{algorithm2e}
\usepackage{circledsteps}
\usepackage{comment}
\usepackage{thm-restate}

\usepackage[normalem]{ulem}

\usepackage{pifont}% http://ctan.org/pkg/pifont

\usepackage{framed}
 {\endMakeFramed}
\definecolor{shadecolor}{gray}{0.75}

\definecolor{lightgreen}{HTML}{008A00}
\definecolor{darkgreen}{HTML}{005700}
\definecolor{lightpurple}{HTML}{e1d5e7}
\definecolor{darkpurple}{HTML}{9673A6}
\definecolor{lightyellow}{HTML}{E3C800}
\definecolor{darkyellow}{HTML}{B09500}

\newcommand{\greencircle}[1]{{\pgfkeys{/csteps/inner color=white}\pgfkeys{/csteps/outer color=darkgreen}\pgfkeys{/csteps/fill color=lightgreen}\Circled{#1}}}
\newcommand{\yellowcircle}[1]{{\pgfkeys{/csteps/inner color=black}\pgfkeys{/csteps/outer color=darkyellow}\pgfkeys{/csteps/fill color=lightyellow}\Circled{#1}}}
\newcommand{\purplecircle}[1]{{\pgfkeys{/csteps/inner color=black}\pgfkeys{/csteps/outer color=darkpurple}\pgfkeys{/csteps/fill color=lightpurple}\Circled{#1}}}

\renewcommand{\ss}{\textsc{Static}\xspace}
\newcommand{\rs}{\textsc{Rotor}\xspace}
\newcommand{\cs}{\textsc{Demand-aware}\xspace}

\newcommand{\rsn}{\emph{rotor-net}\xspace}

\newcommand{\msn}{\emph{mix-net}\xspace}

\newcommand{\csn}{\emph{da-net}\xspace}

\newcommand{\salg}{\ifmmode \mathcal{A}\xspace \else $\mathcal{A}\xspace$ \fi}

\newcommand{\para}[1]{\vspace{2pt} \noindent \textbf{#1:}\xspace}
\newcommand{\rrec}{R_{r}}
\newcommand{\crec}{R_{d}}
\newcommand{\snum}{k_{s}}
\newcommand{\rnum}{k_{r}}
\newcommand{\cnum}{k_{d}}

\newcommand{\dctrot}{DCT_{rot}}
\newcommand{\dctda}{DCT_{da}}
\newcommand{\dctmix}{DCT_{mix}}

\newcommand{\eff}{\eta} 
\newcommand{\talg}{\texttt{GreedyMixNet}\xspace} 
\newcommand{\gmn}{\texttt{GMN}\xspace} 
\newcommand{\bvn}{\texttt{BvN}\xspace}
\newcommand{\val}{\texttt{Val}\xspace}

\newcommand{\system}{{\sc{D3}}\xspace}
\newcommand{\opera}{{\sc{Opera}}\xspace}

\newcommand{\duo}{{\sc{Duo}}\xspace}

\newcommand{\debruijn}{de Bruijn\xspace}
\newcommand{\sirius}{{\sc{Sirius}}\xspace}

\newcommand{\hadoop}{\textsc{Hadoop}\xspace}
\newcommand{\websearch}{\textsc{Websearch}\xspace}
\newcommand{\datamining}{\textsc{Datamining}\xspace}

\newcommand{\simalg}{\mathcal{A}_{\mathrm{hor}}}
\newcommand{\bvnalg}{\mathcal{A}_{\mathrm{ver}}}

\def\gray{\cellcolor{gray!30}}

\newcommand{\card}[1]{\lvert #1 \rvert}

\def\nspace{\vspace{-5pt}}
\newtheorem{observation}{Observation}
\newtheorem{claim}{Claim}
\newtheorem{definition}{Definition}

\newtheorem{theorem}{Theorem}[section]

\newcommand{\revision}[1]{\color{blue}{#1}\color{black}}
\renewcommand{\revision}[1]{#1}
\newcommand{\revisiontwo}[1]{\color{violet}{#1}\color{black}}
\renewcommand{\revisiontwo}[1]{#1}

\renewcommand\footnotetextcopyrightpermission[1]{} % removes footnote with conference info
\setcopyright{none}

\acmDOI{}

\acmISBN{}

\acmConference[]{}
\acmYear{2024}
\begin{document}
\title{D3: An Adaptive Reconfigurable Datacenter Network
}

%\renewcommand{\shortauthors}{Zerwas et al.}

\author{Johannes Zerwas} 
\affiliation{%
  \institution{TUM School of Computation, Information and Technology, \\
  Technical University of Munich}
  \city{Munich}
  \country{Germany}
}

\author{Chen Griner} 
\affiliation{%
  \institution{School of Electrical and Computer Engineering, Ben-Gurion University of the Negev}
  \city{Beer-Sheva}
   \country{Israel}
}

\author{Stefan Schmid}
\affiliation{%
   \institution{TU Berlin \& Fraunhofer SIT}
   \city{Berlin}
   \country{Germany}
}

\author{Chen Avin}
\affiliation{%
   \institution{School of Electrical and Computer Engineering, Ben-Gurion University of the Negev}
   \city{Beer-Sheva}
   \country{Israel}
}

\begin{abstract}
The explosively growing communication traffic in datacenters
imposes increasingly stringent performance requirements on
the underlying networks. Over the last years, researchers have
developed innovative optical switching technologies that enable reconfigurable datacenter networks (RCDNs) which support very fast topology reconfigurations.

This paper presents \system,  
a novel and feasible RDCN architecture that improves throughput and flow completion time. 
 \system quickly and jointly adapts its links and packet scheduling toward the evolving demand, combining both demand-oblivious and demand-aware behaviors when needed. 
\system relies on a decentralized network control plane supporting greedy, integrated-multihop, IP-based routing, allowing to react, quickly and locally, to topological changes without overheads. 
A rack-local synchronization and transport layer further support fast network adjustments. 
\revision{Moreover, we argue that \system can be implemented using
the recently proposed Sirius architecture (SIGCOMM 2020).} 

We report on an extensive empirical evaluation using packet-level simulations. We find that \system improves throughput by up to 15\% and preserves competitive flow completion times compared to the state of the art. 
We further provide an analytical explanation of the superiority of \system, introducing an extension of the well-known Birkhoff-von Neumann decomposition, which may be of independent interest. 
\end{abstract}

\maketitle

\sloppy

\section{Introduction}
Communication traffic in datacenters is growing explosively.
This is due to the popularity of data-centric cloud applications such as batch processing and distributed machine learning (ML), and the trend towards resource disaggregation in datacenters~\cite{talk-about,li2019hpcc}.
This results in increasingly stringent performance requirements on the underlying datacenter networks (DCNs), which are reaching their capacity limits. 
Accordingly, over the last few years, researchers have made great efforts to improve the capacity of DCNs~\cite{clos,jupiter,f10,bcube,mdcube,xpander,jellyfish}.

A particularly innovative approach is to render DCNs dynamic and \emph{reconfigurable}~\cite{osn21}: emerging optical switching technologies allow to quickly change the network topology. In principle, such reconfigurable datacenter networks~(RDCNs) enable us to better use the network capacity by providing topological shortcuts, hence reducing multi-hop forwarding and saving \emph{bandwidth tax} at the cost of \emph{latency tax} \cite{sigmetrics22cerberus}. Indeed, it has been shown that even \emph{demand-oblivious} RDCNs such as RotorNet~\cite{rotornet}, Opera~\cite{opera}, Sirius~\cite{sirius}, and Mars~\cite{sigmetrics23mars}, whose topology changes periodically, can significantly improve the throughput in datacenters. 
\emph{Demand-aware} RDCNs even allow to adapt the topology to account for structure in the workload~\cite{tracecomplexity,benson2010network,facebook,kandula2009nature,mogul2012we,DBLP:journals/cn/ZouW0HCLXH14, datacenter_burstiness}, and, e.g., optimize the rack-to-rack interconnect toward elephant flows. 
Demand-aware networks such as ProjecToR~\cite{projector},
Gemini~\cite{zhang2021gemini},
or Cerberus~\cite{sigmetrics22cerberus}, among  others~\cite{zhou2012mirror,kandula2009flyways,firefly,osa,100times,fleet,flexspander,megaswitch,eclipse,helios,mordia,cthrough,apocs21renets,splaynet,schwartz2019online,dinitz2020scheduling,spaa21rdcn,perf20bmatch,poutievski2022jupiter}, 
successfully exploit the spatial and temporal locality in traffic patterns to improve performance, even if reconfigurations are performed infrequently~\cite{zhang2021gemini,poutievski2022jupiter}. 

\begin{figure}[t]
    \centering
    \hspace{-10pt}
    \includegraphics[width=0.48\textwidth,trim=0cm 0.32cm 0.0cm 0.4cm,clip]{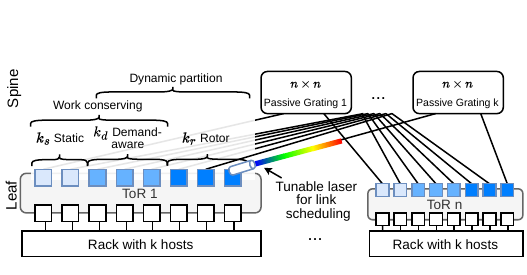}
    \nspace
    \caption{Overview of \system's architecture. The solid lines indicate bidirectional links. ToRs connect in a two layer leaf-spine topology to passive gratings. A tuneable laser at the transceivers can adjust the wavelength to select an egress port. There are three links (ports) scheduler classes \ss, \cs, and \rs.
    }
    \nspace
    \label{fig:trio:architecture}
\end{figure}

Accordingly, the best (dynamic) topology depends on the traffic and applications it needs to serve~\cite{sigmetrics22cerberus}. For example, shuffling traffic of map-reduce applications, or all-gather in ML, will benefit from demand-oblivious RCDNs, based on optical \emph{rotor} switches~\cite{rotornet,opera,sirius}, due to the all-to-all nature of the traffic. At the same time, workloads that are skewed, like reduce-based operations in ML or applications with large flow sizes like in datamining~\cite{greenberg2009vl2} are best served on demand-aware RDCNs, using demand-aware optical switches~\cite{zhang2021gemini,osa,fleet}.
Latency-critical mice flows, on the other hand, are best served on static topologies such as fat-trees or Xpander~\cite{clos,jupiter,f10,xpander},
which do not have latency tax resulting from reconfiguration delays.
Matching traffic patterns to the right topology designs is, hence, critical for performance \cite{sigmetrics22cerberus}. 

While communication patterns in datacenters naturally change over time (e.g., due to increasing loads or evolving applications), it can be difficult in practice to tailor datacenter networks toward their demand.
Our proposed system, \system, is an adaptive network design that enables addressing such an evolution in a practical manner, by dynamically splitting its infrastructure between different modes of operations.

Our paper is motivated by recent optical switching technology introduced in Sirius~\cite{sirius} which enables us to perform topology engineering at the Top-of-the-Rack (ToR). Sirius uses nanosecond tunable lasers at the leaves (ToR switches) and passive gratings that route light based on wavelengths at its spine switches. To change the topology, the optical circuits are reconfigured via the lasers. Controlling the topology boils down to what we call  \emph{link scheduling}, scheduling the wavelengths that each laser uses at any time.
Sirius realizes its RDCN based on a demand-oblivious topology, using demand-oblivious link scheduling. However, as the authors of Sirius already observed, the optical switches can also be used for adaptive link scheduling in a demand-aware manner. 
In fact, we claim that the switch can even be used in a polymorph manner, to quickly change between demand-oblivious, demand-aware, or static link scheduling.
Changing the optical topology, whether demand-oblivious or demand-aware, simply means changing the type of scheduler that controls the tunable lasers. 
Moreover, different lasers (ports), even at the same ToR, can use different types of link schedulers.
This, in principle, enables fully dynamic, \emph{self-adjusting} datacenter networks: networks whose topology can instantaneously be optimized towards the traffic at any time.

\noindent \textbf{Our contribution.}
This paper explores how to design such flexible self-adjusting datacenter networks. In particular, we propose a novel datacenter architecture, \system, based on the Sirius~\cite{sirius} architecture, but whose topology consists of different sub-topology components (namely a static, a dynamic demand-oblivious, and a dynamic demand-aware sub-topology) that can be adjusted quickly toward the evolving traffic. \autoref{fig:trio:architecture} presents \system's architecture, and we explain it in more details in \S \ref{sec:system}. We further study how to support such flexible architectures using efficient and \emph{jointly} optimized link scheduling (to realize the dynamic topology) and packet scheduling (the packet forwarding strategy), as well as network and transport layers. To this end, \system relies on a decentralized network control plane supporting greedy integrated-multihop, IP-based routing, which allows to quickly and locally react to topological changes without overheads and with minimal packet reordering. \system employs different transport protocols for the different sub-topologies: latency-sensitive (small) flows
are transmitted on the static topology using NDP~\cite{ndp}, flows sent via rotors benefit from LocalLB (a local version of RotorLB~\cite{rotornet}), and standard TCP is used for flows transmitted over the demand-aware ports. 

\system comes with theoretical underpinnings, and we explain {\sc{D3}}`s superiority analytically. That is, we prove that a mixture of demand-aware and demand-oblivious scheduling can improve the \emph{demand completion time} of a scheduling that uses only a single type. 
For our analysis, we contribute a novel extension of the Birkhoff–von Neumann (BvN) matrix decomposition which supports mixed topologies and which may be of independent interest.

We further evaluate \system using extensive simulations with realistic and synthetic workloads and find that it improves throughput by up to 15\% while preserving or even improving flow completion times compared to the state-of-the-art. 

We present the main concepts within the paper, and defer some technical details to the Appendix. \emph{The paper raises no ethical concerns.}

\begin{table}[t]
   \footnotesize
    \centering
    \begin{tabular}{|m{2.7cm}|m{1.9cm}|m{2.7cm}|}
    \hline
    Concept (Approach) & Adopted from & \system improvement   \\ \hline \hline
    Demand-aware links \newline (Distributed matching) & ProjecToR~\cite{projector} &  \gray Non-segregated routing \\ \hline
    Rotor scheduling \newline (RotorLB) & RotorNet~\cite{rotornet} &  \gray Local LB \\ \hline
    ToR tunable lasers \newline (Passive gratings) & Sirius~\cite{sirius} & \gray Demand-aware ports \\ \hline
    Three sub-topologies \newline (Static partition) & Cerberus~\cite{sigmetrics22cerberus} &  \gray Dynamic partition  \\ \hline
    Greedy routing \newline (\debruijn topology) & Duo~\cite{zerwas2023duo} & \gray Rotor ports \\ \hline
    Birkhoff–von Neumann \newline (Greedy decomposition) & Eclipse~\cite{eclipse} & \gray Mixed decomposition \\ \hline
    \end{tabular}
     \caption{\system combines and extends several important concepts from existing systems. 
    }
    \nspace
    \nspace
    \label{tab:sota}
\end{table}

\noindent \textbf{Putting things into perspective and related work.}
To achieve its goals, the design and implementation of \system stand on the shoulders of giants 
by combining and extending several important concepts and approaches from existing systems. \autoref{tab:sota} summarizes the main points.
Similar to systems such as ProjecToR~\cite{projector}, \system leverages demand-aware links that are scheduled in a decentralized manner, using distributed matching algorithms; however, \system additionally supports non-segregated routing, enabling an improved utilization of the available network resources. Similar to systems such as RotorNet~\cite{rotornet}, \system leverages fast rotor scheduling to shuffle some of its traffic in a demand-oblivious manner quickly; however, in contrast to the RotorLB approach in prior work,  \system employs a local approach, which significantly reduces control plane overheads and does not require global synchronization.  
As mentioned, we use Sirius \cite{sirius} as the infrastructure for \system, and while the Sirius technology makes \system feasible, the two systems are conceptually different. First and foremost, Sirius is demand-oblivious as a philosophy that trades throughput with simplicity. To achieve higher throughput (as demonstrated in the evaluation section), \system requires a more complex, nontrivial control mechanism to enable demand-aware topologies and multi-hop packet forwarding.
Like Cerberus~\cite{sigmetrics22cerberus}, \system is tailored toward the traffic mix it serves, using a network topology that matches the demand; however, while in the Cerberus topology, the partition into static, demand-oblivious, and demand-aware dynamic switches is \emph{fixed}, \system allows for a dynamic adaption of the topology and its behavior, leveraging a flexible link scheduling (enabled by Sirius technology). 
Moreover, while Cerberus presented a more conceptual contribution and a flow-based simulation, \system, presents a feasible implementation and a packet-level simulation, as well as transport layer solutions.   
\system further builds upon concepts from Duo~\cite{zerwas2023duo}; in particular, it supports greedy routing by relying on an enhanced de Bruijn topology; however, \system additionally also supports rotor ports to shuffle traffic quickly. Last but not least, at its core, \system~relies on a greedy Birkhoff-von Neumann (BvN) decomposition algorithm, as also used in Eclipse~\cite{eclipse}. For \system, we extend this BvN algorithm and introduce a novel matrix decomposition that supports mixed topology types.

\section{Motivation for Links  Scheduling}\label{sec:motivation}

Our work is motivated by the observation that in order to maximize throughput, the datacenter topology must adjust to the specific traffic pattern it serves~\cite{sigmetrics22cerberus}.
We illustrate this with a simple motivating example. Consider an RDCN with five hosts and five racks (single host per rack). 
In this example, a \emph{single reconfigurable port} in each ToR switch provides a single, ToR-to-ToR, directed matching (the reconfigurable topology) between the racks (and hosts). The link scheduler determines the current matching (topology) to use.
The throughput-maximizing topology resp.~link scheduling depends on the ($5 \times 5$) host-to-host demand matrix.

\begin{figure}[t]
\begin{tabular}{cc}
    \raisebox{11pt}{\includegraphics[width=0.08\textwidth]{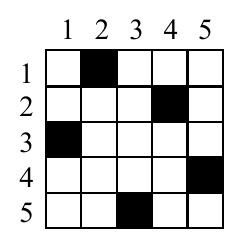}} & 
    \includegraphics[width=0.38\textwidth]{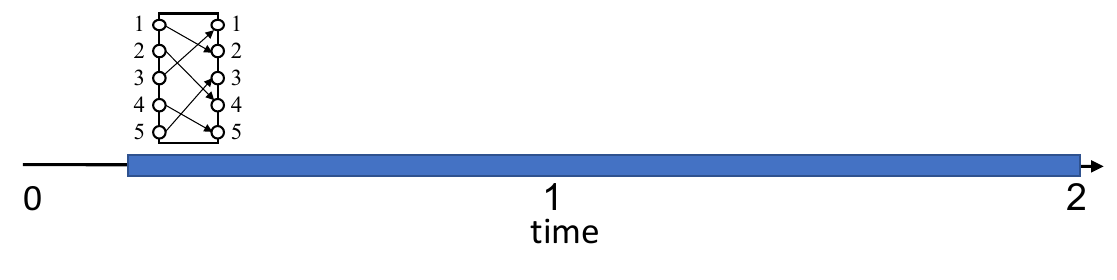} \\
    \multicolumn{2}{c}{(a) Permutation matrix and its optimal link scheduling.} \\
    \raisebox{11pt}{\includegraphics[width=0.08\textwidth]{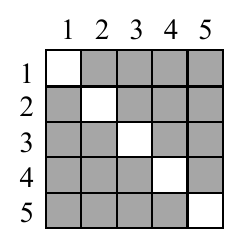}} & 
    \includegraphics[width=0.38\textwidth]{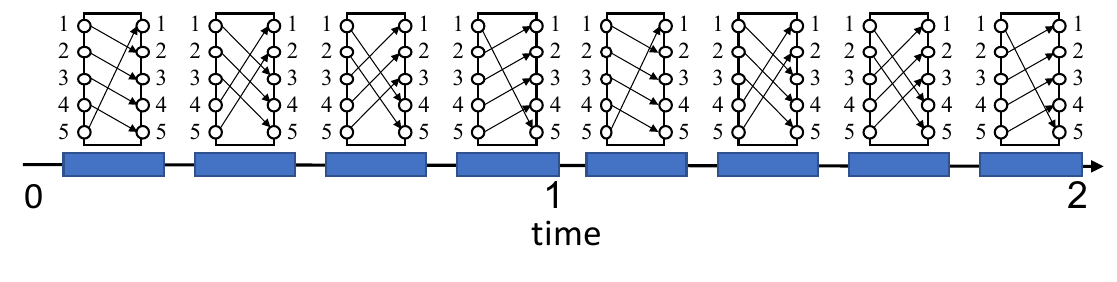} \\
    \multicolumn{2}{c}{(b) Uniform matrix and its optimal link scheduling.} \\
    \raisebox{11pt}{\includegraphics[width=0.08\textwidth]{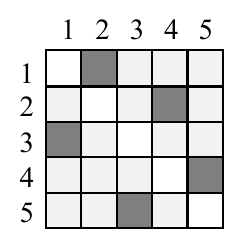}} &
    \includegraphics[width=0.38\textwidth]{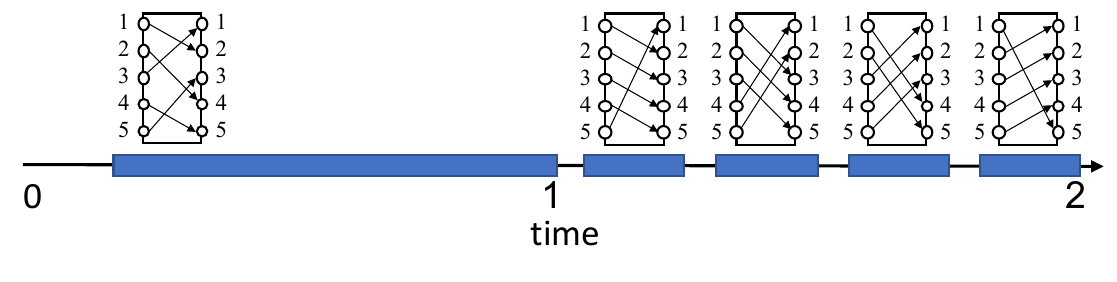} \\
    \multicolumn{2}{c}{(c) Mixed matrix and its optimal link scheduling.}
\end{tabular}
\caption{Motivation for \system: the optimal dynamic topology resp.~ link scheduling depends on the demand matrix, which may change over time.
}
\label{fig:motivation2}
\end{figure}
\autoref{fig:motivation2} presents three different demand matrices: (a) a sparse demand matrix that forms a permutation, (b) a uniform demand matrix where all host pairs have equal demand, and (c) a mixed scenario where darker entries denote more demand.
Next to each matrix, we show its optimal schedules.
The time axis shows two units of time (e.g., 100ms each), and we indicate the time it takes to serve the given demand matrix (the demand-completion time). 
On the time axis, in blue color, we show the circuit-hold times (i.e., the times that circuits can be used to transmit), and the gaps between them represent the reconfiguration times between two consecutive matchings.
Recall that reconfiguration times are typically shorter for demand-oblivious link scheduling (which does not involve any optimization and uses pre-defined matchings) than for demand-aware link scheduling that needs to make decisions, like what links to establish. 
\autoref{fig:motivation2}(a):~A permutation matrix is best served by a single demand-aware matching that provides direct connectivity between the racks that need to communicate. It requires a single links-change (that takes longer) but then makes no more changes. 
\autoref{fig:motivation2}(b):~A uniform demand matrix is best served by a fast and oblivious (periodically) changing sequence of matchings (i.e., a rotor switch), always providing one direct connection between each rack pair during a periodic cycle. 
\autoref{fig:motivation2}(c):~For a mixed demand, a combination of scheduling gives the best results: a demand-aware link scheduling configures a matching for the first unit of time, and then a sequence of fast pre-defined oblivious matchings are used in the second unit of time. We define and show this formally in \S \ref{sec:analysis}.

We conclude that different types of link schedulers (which have a given but different reconfiguration times) are needed to maximize the throughput. Moreover, the \emph{type} of link scheduler on a port may change over time.     
However, realizing such flexible datacenter networks is challenging and requires efficient and practical links and packet schedulers, routing algorithms, and a transport protocol that can support a high degree of dynamics.
In the next section, we describe the design of \system that achieves these goals.

\section{The System Design of \system}\label{sec:system}

This section introduces the system design of \system. We first present the general architecture and topology components. Then, we  describe in more detail the \debruijn-based demand-aware sub-topology and the packet scheduler for the demand-oblivious sub-topology. \revision{We show how integrated forwarding and dynamic port partitioning can be realized in practice and describe the transport layer of \system to fully reap \system's throughput benefits.
Lastly, we discuss how \system can be built on top of Sirius and the implications for cost effectiveness.}

\subsection{Basic Architecture and Link Scheduler}\label{ssec:topology}
\autoref{fig:trio:architecture} overviews \system's architecture. 
\system relies on a two-layer, leaf-spine optical topology that can be described with the ToR-Matching-ToR (TMT) model~\cite{sigmetrics22cerberus}.
The leaf layer consists of $n$ ToR switches $T_i$ ($1\leq i \leq n$).
Each ToR has $k$ up- and $k$ downlink ports of rate $r$. The latter ones connect to end-hosts so that the network contains $h=n\cdot k$ hosts in total.
Since the uplinks are bidirectional, we can further separate them into $k$ unidirectional ingress and $k$ egress ports.
The spine layer consists of $k$ passive gratings\footnote{Originally, $k$ \emph{active} optical spine switches in the TMT model.} connected to the ports at the ToRs.
Tunable lasers in each port of the ToRs adjust the laser wavelength to select the egress port of the grating, and by that changing a topology link, as was proposed and demonstrated in Sirius~\cite{sirius}.
Properly configuring the laser wavelengths at each time slot creates a directed multi-hop topology that is based on a set of $k$-directed matchings between the ToRs, i.e., the $i$-th ports in the ToRs are connected to the $i$-th grating, creating the $i$-th matching (of size $n$). \revisiontwo{This is exactly the structure of  the TMT model. }

In turn, the $k$ matchings (and the $k$ ports in each ToR) are partitioned into three \emph{link scheduling} classes\footnote{Denoted as spine switch types in the TMT model.}: \ss, \rs, and \cs. The sizes of the classes 
are $k_s,k_r,k_d$ respectively, keeping the constraint $k = k_s + k_r + k_d$.
A partition to three classes is implemented by assigning the $k$ ports in each ToR switch to the different link schedulers (\ss, \rs, and \cs)
in a consistent and symmetric way. 
In all ToRs, the same $k_s$ ports use the \ss link scheduler, the same $k_r$ ports use the \rs link scheduler, and the same $k_d$ ports use the \cs link scheduler.
Each of the three link scheduling class creates a different sub-topology that are described in the following:

\para{\textbf{Static sub-topology (\ss)}}
The $\snum$ static ports do not reconfigure the links over time.
The resulting topology can be described as the union
of $\snum$ static matchings. 
The static ports provide basic connectivity between the ToRs and can be used to create $\snum$-regular graphs, such as expander graphs, providing low latency, multi-hop routing for short flows, and control messages.
Specifically, \system relies on \debruijn graphs~\cite{leighton2014introduction} for the \ss topology, which is described later.

\para{\textbf{Demand-aware sub-topology (\cs)}} 
The demand-aware topology consists of a collection of $\cnum$ \cs reconfigurable ports per ToR.
They can flexibly be reconfigured to \emph{any}
possible $\cnum$-regular directed graph
between the ToRs and change it over time.
For example, these ports can be used to create direct connections between ToR pairs with high communication demand.
The pure reconfiguration delay of a link by the lasers has been shown to be in the order of nanoseconds~\cite{sirius}.
However, reconfiguring the \cs ports requires coordination between the ToRs (e.g., for data collection and decision-making). To account for this, we denote by $\crec$ the \emph{reconfiguration delay} of a demand-aware port.
The default value we assume is $\crec=1ms$, but we also evaluate other values in \autoref{sec:evaluation}. 
The circuit-hold time after each reconfiguration is dynamic (not constant) during the operation of a \cs port. 
We require it to be much larger than $\crec$ to achieve a high duty cycle and efficient operation.

\para{\textbf{Rotor-based sub-topology (\rs)}}
The rotor-based topology is formed by the set of $\rnum$ \rs ports per switch.
Each \rs port cycles through $n-1$ predefined matchings, emulating a fully-connected network (i.e., complete graph). 
Every \rs port cycles through the same $n-1$ matchings but using a different time shift. This synchronization provides an average \emph{cycle time} of $\frac{n-1}{\rnum}$ slots to complete a single emulation of a complete graph between all ToRs. The symmetry between the ports' link scheduling facilitates a flexible way to enable or disable a \rs link scheduler on a port since the only difference between the \rs ports is their time shift. 
The slot time of the \rs link scheduling class is defined by a circuit-hold time, denoted as $\delta$, plus a reconfiguration time denoted as $\rrec$, which includes the physical link reconfiguration (wavelength change) and additional quiet time to empty active links. The \emph{duty cycle} $\eff$ is the fraction of time traffic can be sent in a slot (i.e., $\eff = \frac{\delta}{\delta+\rrec}$).
The slot time is tuneable and depends on the reconfiguration time, where a reasonable setup is 
to achieve $\eff >90\%$ as in~\cite{rotornet,opera,sirius}. Here, we assume $\eff\approx98\%$ as in~\cite{opera}.

Since the partition of sub-topologies is determined 
by the link scheduling classes of the ports, but not by the underlying infrastructure, the assignment of ports to link scheduling classes can change over time and, by that, change the sub-topologies sizes. In contrast to previous proposals~\cite{sigmetrics22cerberus}, this adds a new dimension of demand-awareness to \system where $k_s, k_d$, and $k_r$ can dynamically adapt.
For instance, a \cs port can become a \rs port if the traffic pattern has changed such that the \rs sub-topology component is under-provisioned.
While this \emph{dynamic partitioning} of the topology can generally be used for all three classes, \system limits it to the \cs and \rs ports.

The three sub-topologies have different properties for the packet forwarding behavior:
\ss and \cs can utilize work-conserving forwarding (in the sense of store and forward), whereas \rs requires a \emph{packet scheduling} algorithm to transmit packets successfully.
Both parts are introduced in the subsequent sections.

\nspace
\subsection{\ss and \cs Topology}
To implement the sub-topology with work-conserving forwarding, \system augments a \debruijn graph-based static topology built from the \ss ports with dynamic connections (short-cuts) on the \cs ports.
This concept has recently been shown to be a promising candidate for high-throughput topologies~\cite{zerwas2023duo}.
In particular, and in contrast to previous proposals~\cite{projector,sigmetrics22cerberus}, it supports \emph{integrated multi-hop} routing across links from both the \ss and \cs topology parts, i.e., a single packet can traverse both sub-topologies to reach the destination. 
Moreover, the structural properties of the \debruijn graph enable (IP-based) \emph{greedy} routing using small forwarding tables~\cite{zerwas2023duo}.

\begin{figure}[t]
\centering
    \includegraphics[width=0.85\linewidth]{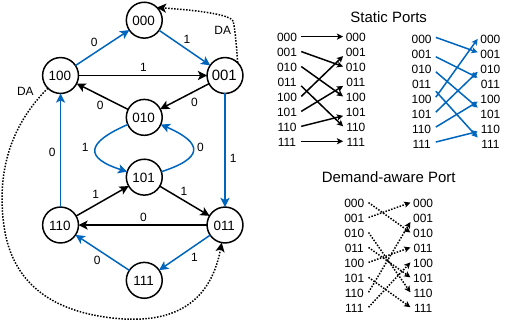} 
    \vspace{-10pt}
    \caption{\revisiontwo{An example of \system's backbone sub-topology with 8 ToRs: A \debruijn topology established by two \ss ports and one \cs port. Top right: two matchings for the \ss ports, bottom right: the matching of the \cs ports (for clarity, drawing only two links in the topology).}}\label{fig:trio:debruijn}
    \vspace{-5pt}
\end{figure}

\autoref{fig:trio:debruijn} shows an example of such a hybrid \debruijn topology with eight ToRs using two \ss ports per ToR, and one \cs port.
Node IDs are in binary representation.
The topology is, therefore, a union of three matchings, two \ss and demand-oblivious, and one dynamic and \cs.
The matchings are shown in the right half of the figure, and we provide the main details below.
Formally, the \emph{static} \debruijn topology is defined as follows~\cite{leighton2014introduction}:

\begin{definition}[\debruijn topology]
For integers $b,d >1$, the {\em $b$-ary \debruijn graph of dimension $d$}, $DB(b,d)$, is a \emph{directed} graph $G=(V,E)$ with $n=\card{V}= b^d$ nodes and $m=\card{E}= b^{d+1}$ directed edges. The node set $V$ is defined as $V=\{v \in [b-1]^d \}$, 
i.e.,  $v=(v_1,\ldots,v_d), v_i \in [b-1]$, 
and the directed edge set $E$ is:
  $ \{v,w\} \in E \Leftrightarrow w \in \{ (v_2,\ldots,v_{d},x): \; x \in [b-1] \}$
where $[i]=\{0, 1, \dots, i\}$.
\end{definition}

The direct neighbors of a node $v$ are determined by a left \emph{shift} operation on the node's address and appending a new symbol $y\in\left[b-1\right]$.
For example, in \autoref{fig:trio:debruijn} where $b=2$ is the binary case, shifting $001$ and appending $0$ to it results in $010$, which is one of two the neighbors of node $001$ in the static sub-topology.
In particular, it is the neighbor via port $0$. The second neighbor is $011$, via port 1, which indicates appending $1$ after the shift operation.
A key property of \debruijn graph-based topologies is that they can be constructed from $b$ directed perfect matchings. Moreover, when using longest prefix matching (LPM), the size of the forwarding table on each node has at most $bd = O(b\log_b n)$ entries~\cite{zerwas2023duo}. 

The essence of the recent proposal Duo~\cite{zerwas2023duo}, was to show that augmenting a \debruijn topology built from $k_s$ \ss matchings, with a $k_d$ \cs matchings creates a topology that supports integrated, multi-hop, greedy, work-conserving, LPM-based routing with a forwarding table size of $O((k_s+k_d) \log_{k_s} n)$ and diameter $d \le \log_{k_s} n$ (Theorem 3.2 within). Moreover, the update cost of a forwarding table upon a \cs link change is small and can be performed locally by communicating with the new neighbor.
\revisiontwo{Routing on the \debruijn topology can be implemented using standard IP and packet switching equipment. More details are given in Appendix~\ref{sec:appendix:debruijn} and~\cite{zerwas2023duo}.}
While Duo has many benefits that we implement in \system, it does not have a \rs sub-topology which we discuss next.

\subsection{Packet Scheduler for \rs Topology}\label{ssec:scheduler}
As in previous work~\cite{sirius, rotornet, opera}, the highly dynamic \rs sub-topology requires a non-work-conserving packet scheduling algorithm to transmit packets successfully. Recall that in this sub-topology, links between ToRs are constantly and systematically changing, emulating a complete graph between the ToRs while being oblivious to the demand.
Therefore, packets potentially need to be stored while waiting for their next-hop link to be reconfigured. 
Like previous work, \system buffers packets at the end hosts (and not at the ToR switch) and forwards them to a ToR switch based on a ToR-host synchronization protocol and the packet scheduler.

The state-of-the-art approach for this challenging task is to consider both \emph{direct} (single hop) and \emph{indirect} (of at most two-hops) routes and use Valiant-based routing~\cite{valiant1982scheme} (via a random intermediate helper node) to achieve load balancing.\footnote{For short, time-sensitive packets, Opera~\cite{opera} uses longer multi-hop paths routing. Still, these packets capture only a small fraction of the total traffic.}
This approach was shown to provide high throughput, i.e., the RotorLB (RLB) scheduling proposed with RotorNet and Opera~\cite{sirius,opera}. RLB has for each host a set of \emph{virtual buffers} for each destination, both for \emph{local} traffic that is generated by the host and for \emph{non-local} traffic where the host acts as an intermediate node. 
Unfortunately, RLB introduces a significant control plane overhead, particularly since intermediate buffers cannot easily handle overflows, and a tight sender-receiver flow control mechanism needs to be implemented.  
At the beginning of every slot, hosts (in different racks) negotiate the amount of traffic that can be sent indirectly to prevent overflow.
In contrast, \system implements a simpler packet scheduler, \emph{LocalLB} (LLB), that does not require such global synchronization. The decisions of what to send indirectly are made (rack-)locally.
In case of an overload of a specific rack/host, i.e., if non-local traffic accumulates, \system can rely on its efficient backbone topology to do offloading and forward long waiting traffic to its destination.
A second difference is the reserved capacity per source. 
\opera assumes a uniform distribution of the traffic and, therefore, a priori, applies an equal share of the slot capacity across all (source) hosts in a rack. 
The initial sending capacity assigned to each host is $\frac{C}{h}$, where $C$ is the slot capacity and $h$ is the number of hosts.
\system does not make such an assumption but allows a more flexible distribution of the resources.

\begin{table}[t]
    \centering
    \footnotesize
    \renewcommand{\arraystretch}{1.05}
    \begin{tabular}{c|c|c|c|c}
         & Offloading & Non-local & Local & Indirect (Sync.) \\\hline\hline
         RotorLB (RLB) & Not supported & FS & FS & FS (global) \\\hline
         LocalLB (LLB) & Non-local & FS & FS & Greedy (local) \\
    \end{tabular}
    \caption{Comparison of RotorLB to \system's LocalLB.}
    \label{fig:trio:scheduling}
    \vspace{-0.7cm}
\end{table}

\autoref{fig:trio:scheduling} summarizes the scheduling process at the beginning of a slot for RLB vs. LLB.
First, with LLB each host checks if non-local traffic demand must be offloaded from the \rs to the backbone sub-topology.
The offloading (\autoref{alg:trio_offloading} in appendix) follows a simple, greedy logic. \revision{Each host $v$ checks for each destination $u$ (that is not in the same rack) if the local demand exceeds a given threshold $d_{\mathrm{off}}$. 
If so, all the non-local traffic to $u$ is offloaded to the \debruijn-based backbone.
If the local demand to $u$ is below the threshold, the algorithm compares the total demand from $v$ to $u$ and offloads excess \emph{non-local} demand to the backbone.}\
Note that LLB offloads only non-local demand to the backbone.

Next, the remaining non-local demand is scheduled. The procedure is the same for RLB and LLB using a fair-share (FS) algorithm over source and destination hosts.
The same procedure is applied for local direct demand.
After this step, RLB creates offers to be exchanged between connected ToRs, thereby synchronizing demand information across the network to perform flow control. In LLB, this step is not necessary, saving complexity.
For allocating new indirect traffic, RLB uses again FS to distribute the remaining capacity across the source hosts and, in particular, the destinations per host.
In contrast, LLB follows a more greedy share (GS) approach since it can handle buffer overloads (see \autoref{alg:trio_scheduling} in appendix).
\revisiontwo{Each host selects the destination with the largest demand and then distributes the remaining capacity in the slot to all hosts with remaining local demand again in a greedy way preferring source destination pairs with small leftover demand and capacity. The detailed algorithm is given in the appendix.}

\subsection{Technical Details}\label{ssec:transport}
\system's architecture with three sub-topologies comes with several challenges regarding integration and synchronization.
In the following, we first cover flow classification, feasible transport protocols, and data and control plane components. Afterward, we describe how packets are forwarded, the synchronization between hosts and ToRs, \revision{and sketch the procedure to re-assign ports to a different scheduler.}

\para{Flow classification}
\system requires a mechanism to classify flows and to separate traffic onto the three sub-topologies. One option is to estimate the flows' sizes to classify them as done in \opera~\cite{opera}, \textsc{Cerberus}~\cite{sigmetrics22cerberus}, and \duo~\cite{zerwas2023duo}. 
Also, \system applies this method to distinguish \ss and \cs traffic, for instance, using approaches such as flow aging or based on information from the application.
However, the \rs sub-topology is particularly suited for uniform traffic patterns, such as flows that belong to the shuffle phase of a map-reduce job.
Therefore, in addition, \system relies on application-level information for traffic classification. 
That is, either a shim layer between the application and network stack on the host inserts tags that identify \rs traffic, or alternatively, transport layer ports are used to identify the applications with uniform communication patterns.

\begin{figure}
    \centering
    \includegraphics[width=\linewidth]{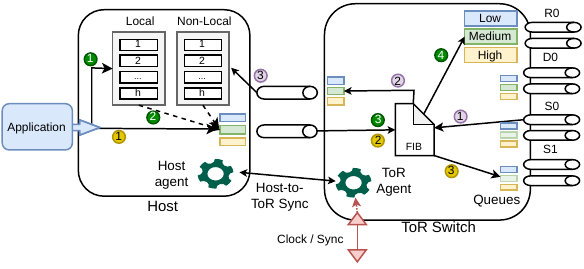}
    \vspace{-0.8cm}
    \caption{Data- and control plane for hosts and ToRs.}
    \label{fig:trio:components}
\end{figure}

\para{Transport protocols}
\system uses different transport protocols for the three sub-topologies.
Latency-sensitive (small) flows run via the \ss topology with NDP~\cite{ndp} which has shown good performance for such traffic.
Flows that are transmitted via the \rs part are sent according to the schedules determined by LocalLB (LLB, \autoref{ssec:scheduler}).
Throughput-sensitive flows use standard TCP for transmission on the \ss and \cs ports to efficiently share resources.

\para{Data and Control Plane Components}
\system uses label-based source routing and priority queuing to forward the flows onto the different topology classes and to turn the decisions of LLB into action. 
\autoref{fig:trio:components} visualizes the involved components on the hosts and the ToRs.
The example considers a \system configuration with one \rs (R0), one \cs (D0), and two \ss ports (S0, S1) on the ToR switch.
All three flow classes share the up-link from host to ToR. 
In order to reduce interference, \system uses priority queues on both hosts and ToR switches. 
\ss flows are given the highest priority, followed by \rs and \cs.
Note that queues for \ss and \cs traffic are not needed on \rs uplink ports, but the queue for \rs traffic is needed on both \ss and \cs ports to enable the \emph{offload} from \rs topology to the \debruijn backbone.
Besides the port queues, there are four more entities involved:
\begin{itemize}[leftmargin=*]
    \item The \emph{Local} and \emph{Non-local} buffers store packets for transmission via the \rs part. Similar to RotorNet and Opera, each of them features a dedicated virtual queue per destination in the network~\cite{rotornet,opera}.
    
    \item \texttt{Host agent} runs on every host and sends the individual
    sizes of the \emph{Local} and \emph{Non-local} buffers to a rack-local coordinator (\texttt{ToR agent}). 
    \revisiontwo{It waits for \emph{pull} messages from the \texttt{ToR agent} containing the volume to be sent from the buffer per destination.}
    When sending a packet, the host adds a \emph{slot label} which either indicates the active \rs matching (to send via the \rs topology) or is $0$ if the packet is to be sent via \ss or \cs ports. (Every \rs port cycles through $n-1$ configurations. 
    The slot label indicates the active configuration for all hosts.)
    \item \texttt{ToR agent} coordinates the \rs packet transmission of all hosts in a rack. It can run on the ToR or on one of the hosts. It receives demand information from the \texttt{Host agents} and runs LLB.
    The outcome is sent to the hosts along with the slot label to be used (\emph{pull} messages). 
    \item \texttt{FIB} is a single forwarding table on the ToR for all three sub-topologies. Besides the destination IP address, it matches the slot label to obtain the egress port. Thereby, forwarding to \rs or \cs and \ss can be differentiated. The entries for the \rs links can be pre-computed and do not change over time (unless the number of \rs ports changes). Forwarding entries for \cs and \ss are updated when local links change.
\end{itemize}

\noindent\textit{Step-by-Step Forwarding Example:}
Packets from applications that should be sent via the \rs (Green), are put to the corresponding \emph{Local} destination queue \greencircle{1}.
If packets are offloaded from \rs to the \ss and \cs sub-topologies, the \texttt{Host agent} takes packets from the queue and directly sends them \greencircle{2}. In this case, the $0$ slot label is added to the packet.
Otherwise, the packets are sent according to the calculated schedule (cf. \autoref{ssec:scheduler}) with the slot label as provided by the \texttt{ToR agent}.
Packets that are sent indirectly via the \rs part are encapsulated with the address of the intermediate host.
On the ToR, the packet is matched in the FIB \greencircle{3} and forwarded using the medium priority queue \greencircle{4}.

Packets from applications to be sent via \ss or \cs (Non-Rotor), are tagged with $0$ as slot label and then sent out to the ToR \yellowcircle{1}. Here, they are again matched in the FIB \yellowcircle{2} and forwarded accordingly \yellowcircle{3}.
Traffic received on the ToR's uplinks is forwarded to the destination host \purplecircle{1} \purplecircle{2}.
When a host receives indirect \rs traffic, it adds the packets to the corresponding \emph{non-local} queue \purplecircle{3}.
Packets received on the final destination (host) are forwarded to the application.
On the ToR, all packets are matched in the FIB, the slot label is popped, and the packet is sent to the egress port.

\para{\rs Synchronization}
\system requires two aspects of synchronization for the \rs sub-topology. First, ToRs need to synchronize their configuration state (slot) globally across the network. This can be achieved using a global (broadcast) clock signal (red triangle in \autoref{fig:trio:components}).

Second, \system needs synchronization between hosts and ToRs to put the decisions of LLB into effect.
Therefore, \system considers a similar approach as RotorNet and Opera.
Packets are primarily buffered on the hosts.
Demand information is pushed \emph{once} per slot from the \texttt{Host agent} to the \texttt{ToR agent}, e.g., with RDMA mesages~\cite{rotornet}. 
After having calculated the number of packets to send with LLB, the \texttt{ToR agent} pulls traffic from the host, i.e., notifies the \texttt{Host agent} about how much to send (e.g., again using RDMA messages).
Other approaches might be possible as well.
For instance, in its rack-based deployment, Sirius can implement the needed local queues with the buffers on the ToRs. 
The authors refer to Credit-based flow control mechanisms as available in InfiniBand to avoid buffer explosions on the ToR. 
However, Sirius negotiates the scheduling almost on a packet-by-packet basis. While this approach reduces the buffer requirements, it comes at the cost of high inter-ToR synchronization effort.

\para{Dynamic Port Partitioning}
\system can dynamically assign ports to different schedulers to adjust the proportions of the sub-topologies.
Like a ``normal'' reconfiguration of a \cs link, the reconfiguration involves changes in the \texttt{FIB}, which can be computed locally by the \texttt{ToR agent} or even prepared offline.
In addition, the \rs scheduler has to reload the new schedule for the sending lasers, which can also be pre-calculated offline.
The reconfigurations must be coordinated globally, e.g., with a centralized controller. 
All ToRs must change at the same time and, to minimize the impact on the \rs packet scheduling, the reconfigurations happen at the end of a \rs slot.
\system considers two possible types of port-to-scheduler reassignments: (1) \cs to \rs  and (2) \rs to \cs.
In (1), the \texttt{ToR agent} first removes the rules with the respective port from the \texttt{FIB} of the ToR and clears the queues belonging to that port (similar to a \cs link reconfiguration).
Then, it removes the port from the \cs scheduler and assigns it to the \rs one by loading the new \rs matchings, updating sending laser schedule and the entries in the \texttt{FIB}.
Finally, it presents the new matchings to the packet-scheduler (LLB).
No updates to the host agent are required since all necessary information for it to operate is included in the pull messages.
Changing a port from \rs to \cs, happens analogously in reversed order.

\para{Practicality and Cost of \system}
Although \system comes with several challenges regarding the integration of the three sub-topologies, all the described aspects in the previous sections illustrate the feasibility of \system. 
Since the Sirius prototype (and simulation code) are not publicly available, we base our observations on a careful study of the paper \cite{sirius} and personal communication with some of the authors. 
We claim that the principles of \system can be implemented and prototyped using the architecture of Sirius. This includes several main components we already discussed: link scheduling (topology reconfiguration) via tunable lasers on both sender and receiver, IP packet forwarding and multi-hop routing, and time synchronization. The control plane can be implemented using \emph{local} rack-based agents and a \emph{global} SDN-based controller that uses the static topology of \system for control messages. As we mentioned earlier, a control plane for the demand-oblivious links (ports) already exists in Sirius, and a control plane for demand-aware links (ports) was already implemented, e.g., in ProjecToR \cite{projector}.  
Moreover, we claim that since \system uses the same hardware as Sirius, the cost analyses of Sirius still hold.
In \cite{sirius}, it was shown that Sirius is cost-effective compared to a classical electrical-based switched network, i.e., Sirius can reach the same throughput at a lower cost. In turn, this claim follows for \system as well since, as we show next, \system improves the throughput of Sirius-like systems, namely demand-oblivious, rotor-based RDCNs.

\nspace
\section{Empirical Evaluation of \system}\label{sec:evaluation}
We evaluate \system with packet-level simulations using \emph{htsim}~\cite{ndp}.
We compare it to two state-of-the-art RDCNs, \revision{a Sirius-like, pure demand-oblivious system, } \opera \cite{opera}, and \revision{ a demand-aware system, } \duo \cite{zerwas2023duo}, across a range of traffic patterns that build on available empirical flowsize distributions and synthetic patterns.
The evaluation focuses on throughput as the main performance metric but also investigates flow completion times and the efficiency of the resource usage.

\para{Settings, Topologies \& Traffic}\label{subsec:Settings}
We consider topologies with $n=64$ ToRs. 
Each ToR has $16$ bi-directional ports, which are equally split into up- and downlinks ($k=8$).
All links have a capacity of $10$~Gbps. 
This results in a total uplink capacity of $5.12$~Tbps, which is comparable to~\cite{opera,sirius,zerwas2023duo}.

\autoref{tab:eval:parameters} (in Appendix) summarizes the used reconfiguration periods. The values are chosen such that the duty cycles in both dynamic topology parts are similar to prior work~\cite{opera,zerwas2023duo}. Unless stated otherwise, the physical reconfiguration delay is $100$~ns for \rs switches and $1$~ms for \cs switches. On the \rs part, we use a $1.7\mu s$ guard period to empty the fibers before reconfiguration.
The evaluation compares the following three systems:

\begin{figure*}[t]
    \centering
    \subfigure[$L=35\%$.]{\label{fig:eval:tp:hm_load:35}\includegraphics[trim=0 0 1.82cm 0,clip,width=0.277\linewidth]{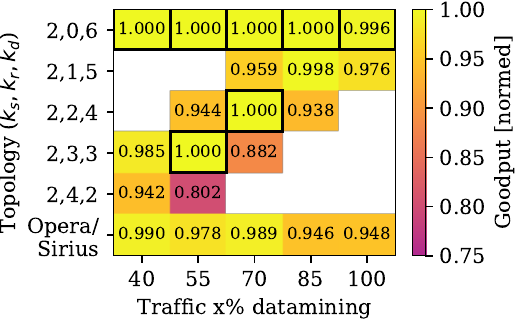}}\hfill
    \subfigure[$L=40\%$.]{\label{fig:eval:tp:hm_load:40}\includegraphics[trim=1.8cm 0 1.82cm 0,clip,width=0.206\linewidth]{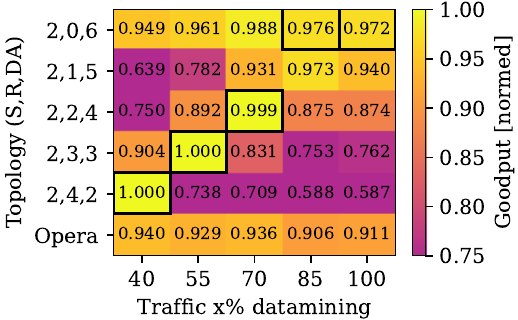}}\hfill
    \subfigure[$L=45\%$.]{\label{fig:eval:tp:hm_load:45}\includegraphics[trim=1.8cm 0 1.82cm 0,clip,width=0.206\linewidth]{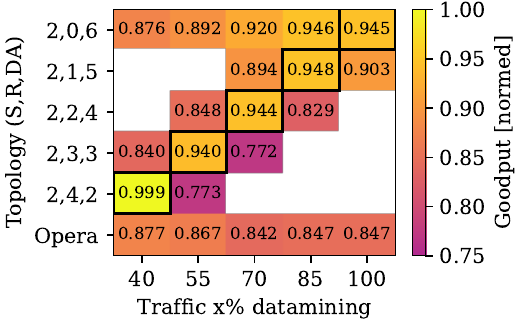}}\hfill
    \subfigure[$L=50\%$.]{\label{fig:eval:tp:hm_load:50}\includegraphics[trim=1.8cm 0 0 0,clip,width=0.277\linewidth]{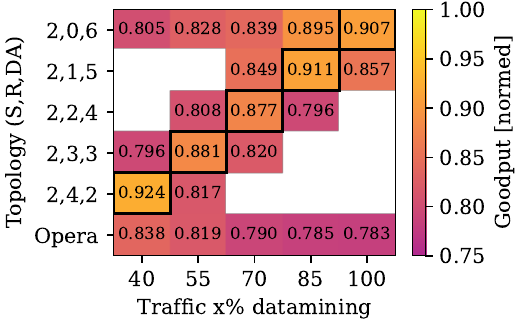}}
    \vspace{-16pt}
    \caption{The average goodput over $1$s of simulation. The heatmap compares topology configurations ($Y$-axis) and traffic shares ($X$-axis) for different loads $L$. Values are normalized to the offered traffic. The best result is highlighted with a thick border. 
    The white regions are excluded for clarity, they are far from optimal.}
    \label{fig:eval:tp:hm_load}
    \vspace{-0.4cm}
\end{figure*}

\para{D3} uses the hybrid topology as presented in \autoref{sec:system}. Throughout the evaluation, we vary the proportions of the parts denoted by the tuple $(k_s,k_r,k_d)$. 
\revisiontwo{Scheduling of \cs links follows the approach for Duo (see below).}
We assume that the system has full knowledge of the flows' sizes at any point in time; an assumption also made in prior work~\cite{opera,zerwas2023duo}.
The allocation of flows to the scheduling classes relies on two criteria. First, we assume application-level information to identify traffic for the \rs part using RLB (see also the traffic description later). 
Second, a size threshold of $1$~MB separates traffic for only the \ss part (using NDP) and traffic for the \cs part (using TCP). 
Lastly, the offloading parameter is $d_{\mathrm{off}}=1$ packet, which is the volume that can be sent in a single slot and a single path, assuming all-to-all traffic and fair-share.

\para{Duo~\cite{zerwas2023duo}} has been presented in prior work, but it can be seen as a specific configuration of \system with $k_r=0$ (no \rs part). 
We consider a configuration with $k_s=2$ and $k_d=6$. The \cs links are scheduled greedily based on the remaining demand volume between the ToR pairs according to Algorithm 2 in~\cite{zerwas2023duo} with a threshold of $10$~MB. The algorithm has full knowledge of the flows' sizes upon their arrival.
The routing uses the properties of \debruijn-based greedy routing, i.e., paths combining \ss and \cs links are possible.
Similar to \system, flows $<1$~MB use NDP and larger flows use TCP.
Queues in the two systems above can hold $50$ data packets of size $1500$~B.

\para{Opera/Sirius~\cite{opera}} \revisiontwo{is a demand-oblivious, dynamic topology and serves as the second baseline. Since an implementation of Sirius is not available, we use Opera as the closest representative of a Sirius-like system.
It periodically cycles through a specifically-generated set of matchings that maintains an expander graph at every time instance. }
\opera also splits the flows in low-latency and bulk traffic. Low-latency traffic is forwarded via the temporarily static expander part of the topology with NDP, whereas bulk traffic is scheduled and sent with the RotorLB protocol and scheduling~\cite{opera,rotornet}.
The evaluation uses their default configuration for queue sizes ($8\cdot 1500$~B) and the threshold for bulk traffic is $15$~MB.

\para{Traffic} We consider an online traffic scenario where flows arrive over time according to a Poisson process.
To demonstrate that \system is particularly suited for traffic patterns that mix skewed demands with uniform patterns, we create traces that contain traffic from two distributions and vary the shares of the patterns. We denote as $x$ the \emph{share} of the skewed traffic.
For the skewed part, connection pairs are sampled uniformly at random.
If not stated otherwise, the flow sizes are sampled from available empirical distributions \datamining~\cite{greenberg2009vl2}.
For the uniform pattern, we create a demand matrix with one flow of size $112.5$~KB for each host pair in different racks. The flows of a matrix arrive over time.
The load $L$ is controlled via the arrival rates of the flows.
We assume that all three systems can identify traffic belonging to the uniform traffic pattern using application-level information and/or special tags. 
They are configured in such a way that they forward flows belonging to the uniform traffic via the \rs part.
We report on $10$ runs per setting.

We first present results about throughput, then about individual flow performance and finally about dynamic traffic.

\subsection{Throughput Evaluation}
We start with evaluating the throughput of \system.

\para{D3 Increases Throughput}
\autoref{fig:eval:tp:hm_load} visualizes the average goodput of $1$~s of simulation for different topology configurations and traffic mixes. 
The values are normalized to the offered load in the respective run. 
A value $=1$ means that the topology can fully serve the offered traffic.
The skewed traffic is sampled from \datamining.
Comparing the different loads (sub-figures), we note that the normalized values decrease slightly with increasing load. 
For $L=35\%$ (\autoref{fig:eval:tp:hm_load:35}), the number of cells $=1$ is highest; whereas for $L=50\%$ (\autoref{fig:eval:tp:hm_load:50}), none of the configurations can fully sustain the offered traffic.
This is expected as the congestion in the topologies increases.

Looking at the individual heatmaps (e.g., \autoref{fig:eval:tp:hm_load:40}), the performance of the topology configurations varies with the traffic share between \datamining and uniform.
That is, for each share, there is one best configuration.
For instance, for $x=70\%$ ($70\%$ \datamining and $30\%$ uniform traffic), \system with $(2,2,4)$ achieves the highest goodput, i.e., the normalized values are closest to $1$.
Moreover, topology configurations with a small number of \rs links serve better traffic mixes with higher $x$.
This aligns with the conclusions made in prior work~\cite{sigmetrics22cerberus}.
The share of \rs links approximately corresponds to the share of uniform-size flows in the traffic.
For instance, the configuration $(2,2,4)$ dedicates $75\%$ of the resources to \ss and \cs and performs best for $x=70\%$ in which $70\%$ of the traffic uses these sub-topologies.
Here, the normalized goodput is $0.999$, compared to $0.892$ and $0.875$ for $x=55\%$ and $x=85\%$ respectively.
To conclude, over-provisioning the \rs part reduces the goodput more than over-provisioning the \cs one.

Compared to \system, \opera performs worse for all considered traffic mixes (and load levels).
For instance, for $L=40\%$, \opera achieves approximately $6\%$ lower goodput than \system.
\duo, which resembles a special case of \system (with $(2,0,6)$), performs the best for the high $x$.
Overall, choosing the correct configuration for \system is important for maximizing its performance.
If not stated otherwise, the following analyses use $L=40\%$ and $x=70\%$ sampled from \datamining.

\para{LocalLB Throughput \& Complexity}
In the previous analyses, \system uses the LocalLB proposed in \autoref{ssec:scheduler}. 
This scheduling can run rack-local, which reduces the synchronization overhead compared to running \system with the RotorLB scheduling~\cite[Algorithm 1]{rotornet}.
\autoref{fig:eval:tp:bar_sched} shows the goodput and the normalized simulation runtime (wall clock time) of both scheduling approaches within \system for $L=40\%$. We acknowledge that the simulation runtime does not reflect the actual synchronization effort, which also depends on hardware characteristics but can give a first indication.
The figure compares two cases: (a) with the topology matching the traffic share, i.e., (2,1,5) is run with $x=85\%$, (2,2,4) with $x=70\%$ etc., and (b) with $x=70\%$.
For all configurations, the differences in goodput between \system and RotorLB are $<10\%$.
For the matching traffic, we observe that LocalLB achieves slightly higher goodput than RotorLB and for the fixed share, the opposite is the case.
LocalLB consistently has lower run times than RotorLB. The difference increases with the number of \rs ports and emphasizes the gain of avoiding inter-ToR synchronization for indirect traffic.

\begin{figure}[t]
    \centering
    \subfigure[Best topology configuration.]{\includegraphics[width=0.49\linewidth]{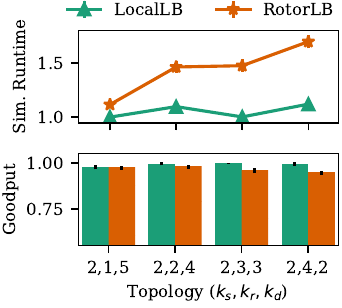}}\hfill
    \subfigure[$70\%$ \datamining.]{\includegraphics[width=0.49\linewidth]{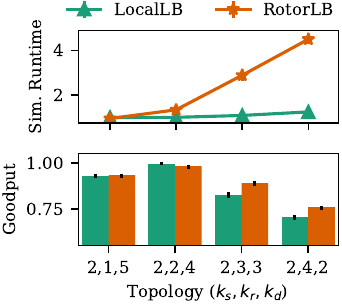}}
    \vspace{-16pt}
    \caption{Comparison of \system's scheduling (LocalLB) and RotorLB scheduling for load $40\%$. }
    \label{fig:eval:tp:bar_sched}
    \vspace{-.6cm}
\end{figure}

\subsection{Sensitivity Analysis}
To demonstrate that \system's advantages also persist in other scenarios, we evaluate its sensitivity to system parameters and traffic characteristics. in the following 

\subsubsection{System parameters}
\label{sec:appendix:sens-param}
We start with assessing the impact of reconfiguration times of the \cs topology, higher link capacities, and varying offloading from \rs to the \cs topology.

\begin{figure*}[t]
    \centering
    \begin{minipage}[t]{0.32\textwidth}
       \includegraphics[width=\linewidth]{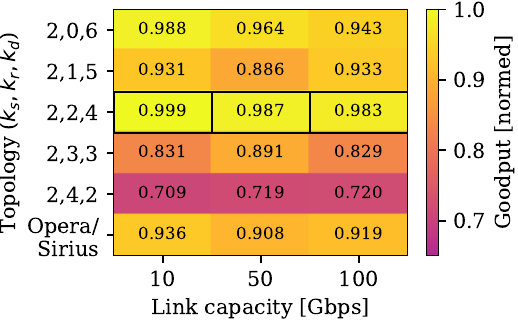}
    	\caption{\revision{Heatmap of the normalized achieved goodput against the link capacity,  and topology configuration.}}
    	\label{fig:eval:tp:hm:link_capa}
    \end{minipage}\hfill
    \begin{minipage}[t]{0.32\textwidth}
        \includegraphics[width=\linewidth]{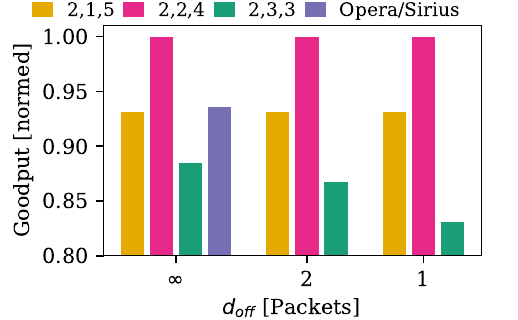}
    	\caption{\revision{Bar of the normalized achieved goodput against the LocalLB offloading threshold ($d_{\mathrm{off}}$) and topology configuration.}}
    	\label{fig:eval:tp:offloading_threshold}
    \end{minipage}\hfill
    \begin{minipage}[t]{0.32\textwidth}
   		\includegraphics[width=\linewidth]{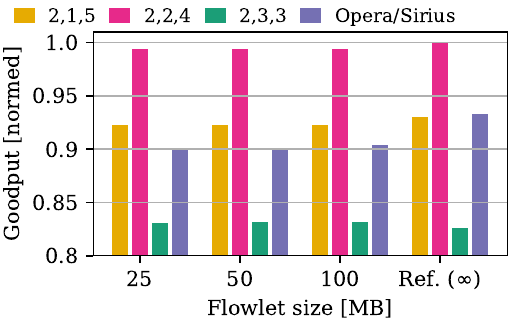}
    	\caption{\revision{Barplots of the mean goodput for traffic with large flows split into chunks  (flowlets).}} \label{fig:eval:burst:fixed_iatf}
    	\label{fig:eval:burst}
    \end{minipage}
\end{figure*}

\begin{figure}[t]
    \centering
    \subfigure[$R_d=10$~ms]{\label{fig:eval:tp:hm_reconf:10ms}\includegraphics[trim=0 0 1.83cm 0,clip,width=0.48\linewidth]{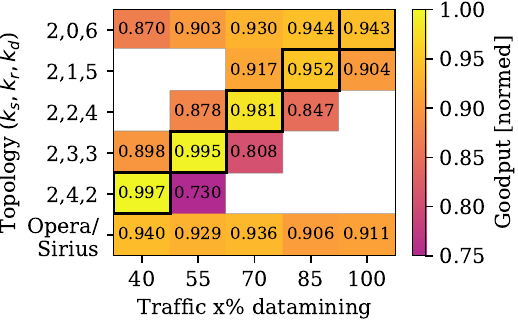}}\hfill
    \subfigure[$R_d=100\mu$s]{\label{fig:eval:tp:hm_reconf:100us}\includegraphics[trim=1.8cm 0 0cm 0,clip,width=0.48\linewidth]{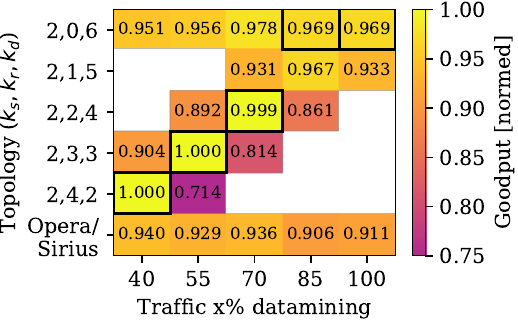}}
    \caption{Goodput averaged over $1$s of simulation. Comparison of reconfiguration times of the \cs links with fixed duty cycle. The heatmap compares topology configurations and traffic mixes.}
    \label{fig:eval:tp:hm_reconf}
\end{figure}

\para{Reconfiguration Times}
\autoref{fig:eval:tp:hm_reconf} illustrates similar heatmaps as in \autoref{fig:eval:tp:hm_load} for $R_d=10$~ms (\autoref{fig:eval:tp:hm_reconf:10ms}) and $R_d=100\mu s$ (\autoref{fig:eval:tp:hm_reconf:100us}). 
The reconfiguration period is adjusted so that the duty cycle remains constant.
The load is $L=40\%$.
On a macroscopic level, we observe a similar behavior as before: matching the topology configuration to the traffic mix maximizes the achieved goodput.
Comparing the values in detail, no consistent impact of the reduction of $R_d$ (and the implied reduction of the reconfiguration period) is observable. Some combinations have higher, some have lower achieved goodput. 
\revision{Note that the values for \opera are not affected when changing $R_d$.}

\para{Link Capacity \& Offloading}
To illustrate that \system's benefits also persist for higher link capacities, \autoref{fig:eval:tp:hm:link_capa} shows the achieved goodput for scenarios with higher link capacity. The detailed parameters are summarized in Tables~\ref{tab:eval:parameters_100G} and~\ref{tab:eval:parameters_50G}. The load is $L=40\%$ with the share of \datamining$x=70\%$.

Overall, we can observe a similar behavior as for 10G. The matching topology achieves the highest throughput.
The performance deviations of \system between 10G and 100G are within the range of single-digit percentages, but there is no consistent pattern.
50G shows slightly different behavior than 10G and 100G for the configurations close to the matching one. Here, the goodput reduces almost symmetrically for too few or too many \rs links, whereas it is asymmetrically for 10G and 100G. We leave a more detailed investigation for future work.
In summary, we conclude that the gains of \system persist also for other link capacities.

\autoref{fig:eval:tp:offloading_threshold} visualizes the impact of the offloading threshold. 
The link capacity is 10G. 
The figure compares three values 1, 2, and $\infty$ and also shows \opera as a reference. The values are normalized to the best configuration with $d_{\mathrm{off}}=1$ (2,2,4).
The achieved goodput varies with the threshold. However, the ranking among the topology configurations persists.
The specific impact depends on the topology configuration. For instance for (2,2,4), $d_{\mathrm{off}}=2$ results in slightly lower goodput than $\infty$ or $1$. In contrast, for (2,1,5), the impact is negligible and for (2,3,3), we observe that the goodput decreases significantly as the threshold is reduced and more traffic is offloaded.
Overall, the amount of offloading should be optimized.

\subsubsection{Traffic pattern}
\label{sec:appendix:sens-traffic}
In the following, we assess if \system outperforms the reference solutions also with other loads, flow size distributions, less bursty traffic and evaluate its robustness to delayed or erroneous flow size information.

\begin{figure*}[t]
    \centering
    \subfigure[Load (Ref: \opera/Sirius).]{\label{fig:eval:gain:load}\includegraphics[width=0.3\linewidth]{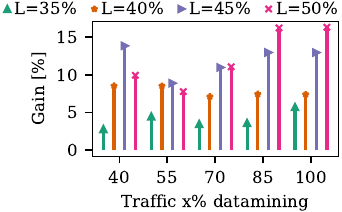}}\hfill
    \subfigure[Distribution (Ref: \duo).]{\label{fig:eval:gain:distr_duo}\includegraphics[width=0.3\linewidth]{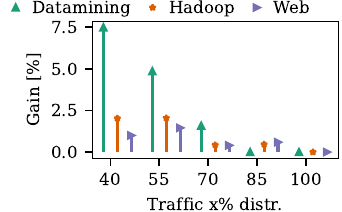}}\hfill
    \subfigure[Distribution (Ref: \opera/Sirius).]{\label{fig:eval:gain:distr_opera}\includegraphics[width=0.3\linewidth]{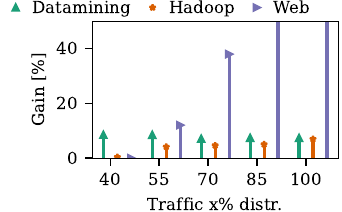}} 
    \vspace{-10pt}
    \caption{\revision{Comparison of gain in goodput when using \system for different flowsize distributions for the skewed traffic (a) and over the offered load (b \& c). The reference topology is indicated in the caption.
    }}
    \label{fig:eval:gain}
\end{figure*}

\para{Performance for Other Loads \& Distributions}
\autoref{fig:eval:gain} summarizes the relative gain in goodput when using \system compared to the indicated reference configuration. 
The gain is calculated as $\frac{G_{\system}-G_{ref}}{G_{ref}}$. Note that (2,0,6) (\duo) is also considered as part of \system here.
\autoref{fig:eval:gain:load} illustrates the gain against the share of skewed traffic for different loads.
For $L=35\%$, the gain varies between $0-4\%$ but does not show a clear trend.
For $40\%\leq L \leq 50\%$, the behavior is different. 
For these load values and $40\%\leq x\leq70\%$, the gain remains almost constant around $10\%$ -- neither the load nor the share of skewed traffic shows an impact here.

For a share $\geq 85\%$, the load clearly impacts the gain. It raises from $\approx 10\%$ at $L=40\%$ to $\approx 16\%$ at $L=50\%$. 
In summary, choosing the correct topology configuration becomes more critical with higher traffic skew and load.

\autoref{fig:eval:gain:distr_duo} and~\ref{fig:eval:gain:distr_opera} illustrate the impact of the flow size distribution of the skewed traffic on the gain.
In addition to \datamining, they also show the results when using \hadoop~\cite{facebook} and \websearch~\cite{alizadeh2010data}.
The reference configurations are \duo (2,0,6) and \opera, respectively.
First, for \duo (\autoref{fig:eval:gain:distr_duo}), we observe that the gain diminishes for all distributions when $x$ increases.
For higher shares, the optimal configuration has smaller $k_r$ and becomes more similar to (2,0,6). 
In fact, for $x=100\%$, (2,0,6) is the best configuration.

Compared to \opera (\autoref{fig:eval:gain:distr_opera}), \system shows improvements up to $10\%$ when skewed traffic is sampled from the \datamining or \hadoop distribution.
The numbers vary depending on the exact share.
For \websearch, the gain strongly grows with $x$ (the upper-end is cut. The max gain is $112\%$ at $100\%$ \websearch).
Also, in prior work~\cite{opera, zerwas2023duo}, \opera's throughput was shown to collapse if the amount of traffic routed via the expander part increases, as is the case with the \websearch distribution.
As a general trend, we observe that the gains decrease with the average flow size of the distributions (\datamining has the highest, followed by \hadoop and then \websearch). 
The intuition is that \hadoop and \websearch are per se less skewed than \datamining, which overall benefits \opera.
In summary, the gains from \system persist for other skewed traffic distributions.

\para{Impact of Reduced Burstiness}
In the previous evaluations, the entire volume of the flow is available for transmission upon flow arrival which results in highly bursty traffic when large flows arrive.
To assess the performance of \system in a less bursty scenario, we split large flows into several flowlets (of a given size) arriving sequentially with a certain delay. 
The flowlet inter arrival time is calculated from the ideal transmission time. The values are normalized to the best solution (2,2,4) in the reference case in which the flowlet size is $\infty$, i.e., the entire flow arrives at once.
\autoref{fig:eval:burst:fixed_iatf}  varies the size of the flowlets.
When matching the topology to the traffic almost no impact of the flowlet size is observervable. Also, the ranking between the topology configurations persists.
For the not-matched configurations (2,1,5) and (2,3,3), there are slight variations in the achieved goodput. 
Also for \opera, the goodput decreases with the flowlet size.

\begin{figure}[t]
    \centering
    \includegraphics[width=0.65\linewidth,trim=0cm 5cmcm 2.3cm 0cm,clip]{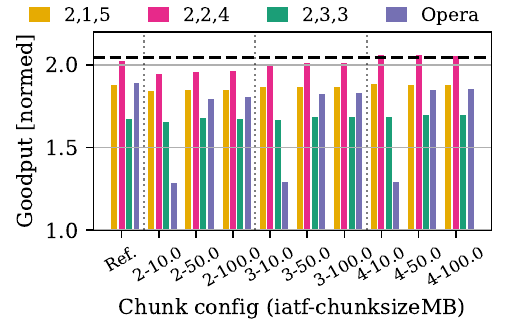}\vspace{-0.2cm}
    \subfigure[Fixed error $=0\%$.]{\label{fig:eval:flowsize_delay:delay}\includegraphics[width=0.49\linewidth]{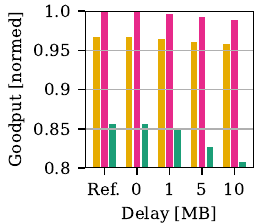}}
    \subfigure[Fixed delay $=1$MB.]{\label{fig:eval:flowsize_delay:error}\includegraphics[width=0.49\linewidth]{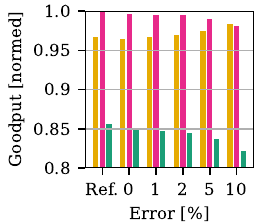}}\hfill
    \caption{\revision{Barplots of the mean throughput for traffic with relaxed assumption about flowsize information. (a) varies the delay of the information and (b) varies the error.}}
    \label{fig:eval:flowsize_delay}
\end{figure}

\para{Classification Errors and Delayed Information}\label{sec:appendix:flowsize}
\system builds on the assumption that it has information about the flow size readily available to classify flows. 
In order to demonstrate its robustness, we relax this assumption along two dimensions: 1) delaying information about the total flow size, 2) introducing errors to the flow classification.

For the first case, \autoref{fig:eval:flowsize_delay:delay}, we delay the flowsize information in the simulation so that the first part of the flow is handled as a small flow.
The full information flowsize is then available after the time it would take to transmit the first part over an ideal link. This can also be interpreted as some kind of flow aging. 
When the full flowsize information is released (emulated as a second flow arrival), the flow can be properly classified and forwarded on the fitting sub-topology.
The figure compares delays of $1$, $5$, and $10$MB which correspond to $667$, $3334$ and $6667$ packets respectively.

The figure shows the results without flow classification errors. Note that since the traffic generation for this analysis is not exactly the same as before, attention should be put on the relative differences and a 1:1 comparison of the numeric values to other figures is not meaningful.
Comparing the results to the reference case, the goodput reduces slightly as the flow size information delay increases. 
This is expected as more traffic is routed via the static topology part. However, the reduction amounts to only $2-3\%$ for $10$MB delay (which is already larger than the majority of the flows). 

The second relaxation (\autoref{fig:eval:flowsize_delay:error}),  is an additional error in the flow classification (on top of a fixed flowsize information delay of 1MB). 
Here, $y\%$ of the uniform flows for the Rotor-topology are mis-classified and forwarded via the static topology part. 
The goodput of the best topology configuration for $x=70\%$ (2,2,4) continuously reduces as the error increases. The reason is that the actual (allocated) traffic shares changes. The share of traffic allocated to the \ss and \cs topology increases.
Therefore, we also observe that (2,1,5), which dedicates more resources to those topology parts, shows increasing goodput, up to the point that the ranking changes at $10\%$ error.
Lastly, the goodput of (2,3,3) reduces for the very same reason as described above.

\nspace
\subsection{Individual Flows' Performance}
This subsection evaluates the impact on the individual flows.

\para{Flow Completion Times (FCT)}
\autoref{fig:eval:fct} visualizes the FCTs against the flow sizes for \system (2,2,4), \duo (2,0,6) and \opera. 
The flows are separated into three groups for the sake of readability. The dashed lines indicate the optimal transmission time accounting for queues and propagation time (using the calculations provided by~\cite{opera}).
For small flows ($\leq100$~KB, \autoref{fig:eval:fct:small}), we consider the $99\%$-ile to understand better how this latency-sensitive traffic is handled.
For all sizes, \system achieves the smallest FCT, followed by \opera and \duo.
Recall that small flows use only the \ss topology. 
Compared to \system, \duo's performance suffers from the uniform traffic that cannot be served well over the \cs links and creates congestion on the \ss links.

\begin{figure*}[t]
    \centering
    \subfigure[Small flows.]{\label{fig:eval:fct:small}\includegraphics[width=0.31\linewidth]{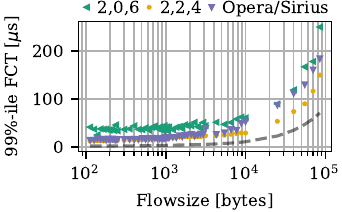}}\hfill
    \subfigure[Medium flows]{\label{fig:eval:fct:medium}\includegraphics[width=0.31\linewidth]{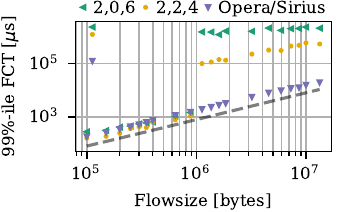}}\hfill
    \subfigure[Large flows.]{\label{fig:eval:fct:large}\includegraphics[width=0.31\linewidth]{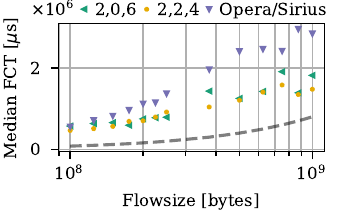}}\hfill
    \vspace{-13pt}
    \caption{Comparison of FCT. Flows separated into three groups. 99\%-ile for small flows (a), medium flows (b),
    and the median for large flows (c). Load is $40\%$, traffic mix with $x=70\%$.
    }
    \vspace{-.4cm}
    \label{fig:eval:fct}
\end{figure*}

For medium sized flows ($99\%$-ile, \autoref{fig:eval:fct:medium}), \system consistently performs better than \duo for the same reason as for the small flows. 
\system also outperforms \opera until the point where \system starts to classify flows for the \cs part ($1$~MB). From thereon, flows are using TCP, and there is a strong increase in the FCT, an effect that has already been observed in previous work, e.g., in \duo~\cite{zerwas2023duo} and similarly also for \opera~\cite{opera} (for flows $>15$~MB).
A special point of interest are flows of size $112.5$~KB, which arrive according to uniform matrices. 
All systems show an increased FCT here. \opera is the best, followed by \system (with the optimized configuration) and \duo (2,0,6).
This steep increase can be explained by the fact that these flows do not fit in a single \rs slot and need multiple cycles of the \rs sub-topology to be transmitted.

The situation changes for the FCTs of large flows ($\geq 100$~MB, \autoref{fig:eval:fct:large}). Since for large flows, throughput matters more, we consider here the median value. \system (2,2,4) performs best followed by \duo (2,0,6) and \opera. Specifically, the distance between \system and \duo is smaller compared to \opera; in fact, the FCTs for \duo and \system overlap for most flow sizes. This relates to our previous observations regarding the achieved goodput.
In summary, integrating \rs into a purely demand-aware topology improves the FCT for all flow sizes. Compared to \opera, \system trades off benefits for the small and large flows against those of medium-sized and the \rs flows.

\para{Packet Reordering}
To take a different perspective on the impact on the individual flows, \autoref{fig:eval:fct:seqno} illustrates the difference between the expected and received sequence number of the packets at the receive.
Thereby, it hints at the packet reordering experienced in the different topologies. 
Again, $L=40\%$ and $70\%$ \datamining.
The abscissae are divided into values $<0$ (useless re-transmissions), $=0$, and more fine-grained for values $>0$ (reception of re-ordered packets).
Comparing the systems, we observe that the behavior of (2,0,6) and \system (2,2,4) is similar.
For both, $>80\%$ of the packets arrive in order. This is fundamentally different for \opera for which $<20\%$ of the packets arrive in order. The shape of the remaining curve has a break around $79$ packets which corresponds to the number of packets needed to transmit the flows from the uniform traffic (of size $112.5$~kB).
We conclude that \opera requires more effort at the receivers to forward packets in order to the applications.

\nspace
\subsection{Dynamic Traffic \& Partitioning}\label{sec:eval:dynamic}
As a last aspect, we assess the benefits of dynamic partitioning in \system. \autoref{fig:dynamic:eval:tp} shows the goodput over time and the average of a simulation with varying traffic mix. Specifically, $x$ increases from $55\%$ for $<2$s to $70\%$ ($2$s $\leq t <4$s) to $85\%$ ($4$s $\leq t <6$s) and finally reduces to $70\%$ again for $\geq6$s.
For each of the three traffic mixes, the figure shows the goodput of the matching configuration (2,3,3), (2,2,4) and (2,1,5) respectively.
For these three static baselines, we see in the period with matching traffic that they achieve the offered load (after some transient phase) but drop in goodput when not matching the traffic. For instance, (2,3,3) significantly drops after $2$s.
In contrast, for the dynamic partitioning, the achieved goodput aligns with the offered load throughout the whole run. This also becomes evident from the average values. 
Note that the shown example is hard-coded and only demonstrates the benefits of applying the dynamic partitioning.

\autoref{fig:eval:dynamic:retransmissions} illustrates the rate of retransmissions captured at the senders throughout the simulation run with changing traffic mix.
The retransmissions are split by the transport protocol (TCP or NDP) and normalized with the total successfully received volume at the destination per $1$ms time window. 
By design, there are no retransmissions for RLB.
The abscissa of the figure is divided to zoom-in to the time ranges close to a change in the traffic pattern, showing $\pm 0.5$s around the change.
All topology configurations show retransmission rates in a similar range.
For NDP, the are no retransmissions.
Moreover, in all cases, we observe periodic spikes due to reconfigurations of the \cs links. During these reconfigurations, all packets in the affected queues are dropped, resulting in retransmissions.
The reconfiguration of the sub-topologies' sizes (at $t=2, 4, 6$s) uses the same approach when converting a \cs port to a \rs one and vice versa (cf. \autoref{ssec:transport}), hence, we observe the effects from \cs reconfiguration also during port conversion.
Besides that, no additional impact is observed hinting at the feasibility of port conversions in \system.
Note that the reconfigurations are scheduled such that they happen between two slots on the rotor topology; therefore, we also do not expect major impacts there.

\begin{figure*}[t]
\centering
\begin{minipage}[t]{0.3\textwidth}
        \includegraphics[width=0.98\linewidth]{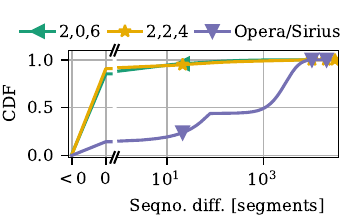}
        \caption{Difference in sequence number of received packets. $40\%$ load, $70\%$ \datamining.}
        \label{fig:eval:fct:seqno}
    \end{minipage}\hfill
    \begin{minipage}[t]{0.36\textwidth}
        \includegraphics[width=0.98\linewidth]{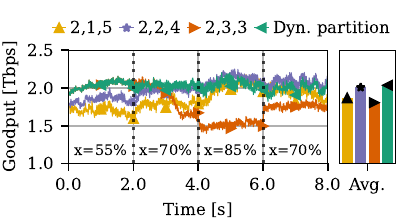}
        \caption{Goodput over time and average. Dashed lines indicate changes in the traffic mix.}
        \label{fig:dynamic:eval:tp}
    \end{minipage}\hfill
    \begin{minipage}[t]{0.3\textwidth}
         \includegraphics[width=0.98\linewidth]{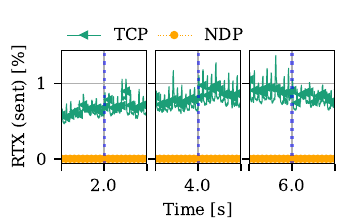}
        \caption{Retransmission rate  per transport protocol at the sender. relative to the received volume.}\label{fig:eval:dynamic:retransmissions}
    \end{minipage}
\end{figure*}

\section{Theoretical Modeling}\label{sec:analysis}

This section provides an analytical explanation for the superiority of \system over \opera/\sirius and \duo in terms of throughput. To this end, we study the \emph{demand completion times} (DCT) of these three systems in a simplified model that captures the essence of the systems. In~\cite{sigmetrics22cerberus}, a direct connection between the DCT, the time to send all the traffic of a demand matrix, and the throughput was proved; basically, it was shown that a shorter DCT implies a higher throughput. This led to a novel method that allows us to account for reconfiguration times when analyzing the throughput of dynamic networks. 

We compare the three systems denoted as i) \csn where all traffic needs to be sent via the \cs link scheduler (matchings) with a reconfiguration delay of $\crec$.
ii) \rsn where all traffic needs to be sent via the \rs link scheduler
that rotates between $n-1$ predefined (oblivious) matchings with a reconfiguration delay of $\rrec$.
iii) \msn a simplified \system-based system where traffic can be split (in time) between the \cs link scheduler and the \rs link scheduler.
We assume all systems have a single dynamic port in the ToRs, namely a single matching, but the topology scheduling can change over time.
Moreover, all systems send the same fraction of the traffic that belongs to time-sensitive flows toward \ss links for multi-hop routing and, therefore, remove it from our analysis.    

Following previous work~\cite{jyothi2016measuring,namyar2021throughput,eclipse}, we model the traffic as a \emph{saturated} demand matrix $M$, i.e., a scaled doubly stochastic matrix where the sum of each row and column is equal to the nodes capacity (depending on the links rates and the time-frame that the matrix captures).    
In our analysis of the DCT, we assume that delay is only due to transmission and reconfiguration times, and we neglect any delays due to packet loss or congestion. 
We now consider the DCT for \csn, \rsn, and the simplified \system in turn. 

\para{DCT of \csn and BvN Decomposition}
Similar to previous work~\cite{eclipse,chang2000birkhoff}, we model the DCT of \csn using the Birkhoff-von Neumann decomposition~\cite{birkhoff1946tres} and single-hop forwarding.
Briefly, the Birkhoff-von Neumann (BvN) theorem guarantees that for any doubly stochastic matrix, $M$, there exists a matrix decomposition to a set $P(M) =\{P_1,\dots,P_m\}$ of $m$ different permutation matrices (each row and each column has exactly a single 1 entry and all other entries are 0), each with a corresponding non-negative coefficient $ \alpha =\{\alpha_1,\dots\alpha_m\}$. The sum of the coefficients stratifies the condition $\sum^m_{i=1} \alpha_i=1$ when the matrix is doubly stochastic or $\card{M}/n$ when the matrix is saturated and $\card{M}$ it the sum of all entries in the $M$. 
The linear combination of these creates the original matrix that is: $\sum^m_{i=1} \alpha_i P_i=M$.
Given that the \cs link scheduler can emulate any matching, we can individually transmit each matching in $P$. 
Transmitting each $P_i$ takes $\frac{\alpha_i}{r}$ plus the reconfiguration time $\crec$. The DCT for \csn, a demand matrix $M$, an algorithm $\bvn$ that computes a decomposition $\mathcal{P}(M)$, and a set of coefficients $\alpha$, is the sum of transmission and reconfiguration times of each permutation: 
\nspace
\begin{align}\label{eq:daDCT}
    \dctda(M, \bvn) &=\sum_{i=1}^m (\frac{\alpha_i}{r}+\crec) =\frac{\card{M}}{nr}+m\crec.
\end{align} 
While it is simple to find a decomposition of an order of $O(n^2)$ terms, the problem of minimizing the number of terms in the decomposition is known to be NP-complete~\cite{dufosse2016notes}. Note that a minimal number of terms will also minimize the total reconfiguration time and, therefore, the DCT. 

\para{The DCT of \rsn}
Recall that \rsn uses faster reconfigurations to rotate between $n-1$ predefined matchings, which, when combined in a full cycle, form a complete graph~\cite{rotornet}. During each slot, there is a single active matching 
on which packets are sent (recall that for now, we have a single \rs link scheduler or port).
Each source-destination pair has to wait for all other $n-2$ matchings to rotate until it is available again.

In practice \rsn and similar projects such as Sirius~\cite{sirius} are demand oblivious and use Valiant load balancing~\cite{valiant1982scheme} routing on \emph{all} packets.
In this case, the fraction of direct single-hop traffic will always be $0$ \revision{and all traffic is sent via two hops.} 
Therefore, to approximate the expected DCT for \rsn in the following sections, we assume that it uses Valiant routing, and we denote it as $\val$.
That is, for a demand matrix $M$ its expected DCT on \rsn is 
\nspace
\begin{align}\label{eq:rsnDCT}
   \dctrot (M, \val) = \frac{2}{\eff r} \frac{\card{M}}{n}.  
\end{align}
\revision{where $\eff= \frac{\delta}{ \delta+\rrec}$ is the \emph{duty cycle}. This bound was previously informally ~\cite{rotornet,sirius} and formally ~\cite{sigmetrics22cerberus} claimed.}

\para{DCT of \msn and BvN Based Partition}
In the above, we provided bounds for two extreme cases: either the whole traffic, $M$ is sent via \csn using only \cs link scheduling, or the whole traffic is sent via \rsn using only \rs link scheduling.  

The basic idea of \system and here \msn, is that a partition of the traffic in $M$ into two parts may lead to better DCT results, namely to split $M$ into one sub-matrix that will be sent using \cs link scheduling and a second sub-matrix that will be sent using \rs link scheduling. 

We extend the classic BvN problem to the problem of DCT minimization on \msn. 
As before, let $M$ be a saturated demand matrix. 
We can use the BvN theorem to decompose $M$ into a set of at most $m \le n^2$ scaled permutation matrices $\{M_1, M_2 \dots M_{m}\}$ where $M_i = \alpha_i \cdot P_i$.
Since each $M_i$ is a scaled permutation matrix, the sum of any subset of matrices is a scaled double-stochastic matrix. 
An algorithm $\salg$ for \msn needs to perform three operations:
\begin{enumerate}[leftmargin=*]
    \item Decompose $M$ into the set $P$ of permutation matrices.
    
    \item Define the subset of matrices $P^{(da)} \subseteq P$ that will be served using \cs link scheduling. Let
    $P^{(rot)} = P \setminus P^{(da)}$ be the set that will be served using \rs link scheduling. We define $M^{(da)}=\sum_{P_i \in P^{(da)}} M_i$ and $M^{(rot)}=\sum_{P_i \in P^{(rot)}} M_i$. Note that
    $M = M^{(da)} + M^{(da)}$.

    \item Define the packet scheduling for $M^{(da)}$ using the \cs link scheduling and the packet scheduling for $M^{(rot)}$ using the \rs link scheduling.
\end{enumerate}
\noindent Given such an algorithm $\salg$, the DCT for \msn is
\begin{align}
    \dctmix(M, \salg) = \dctda(M^{(da)}) + \dctrot(M^{(rot)})
    \label{eq:dctmix}
\end{align}

An optimal algorithm for \msn is an algorithm that minimizes the overall DCT by finding the best decomposition and best partition for $M^{(da)}$ and $M^{(rot)}$. The above definition can be extended to non-saturated matrices (for $M, M^{(da)}$ and $M^{(rot)}$), but for simplicity we keep them saturated for now.  

\begin{figure}[t]
    \hspace{-15pt}
    \begin{tabular}{ccc}
    \multicolumn{3}{c}{\includegraphics[width=0.65\linewidth]{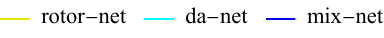}}  \\
    \includegraphics[width=0.3\linewidth]{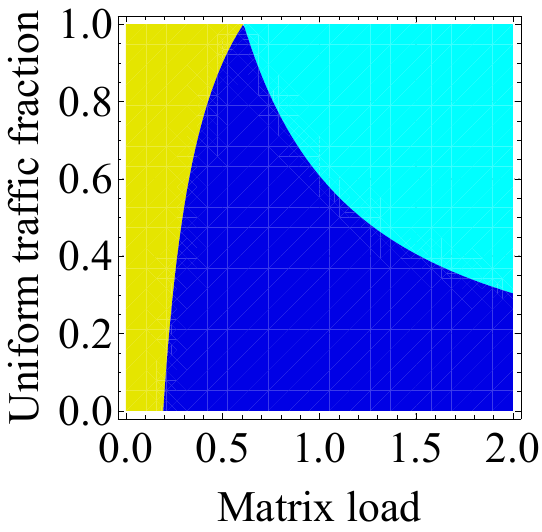} &
    \includegraphics[width=0.3\linewidth]{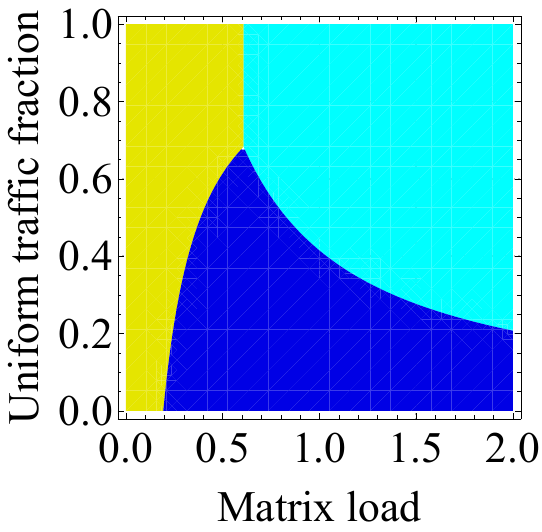} &
    \includegraphics[width=0.3\linewidth]{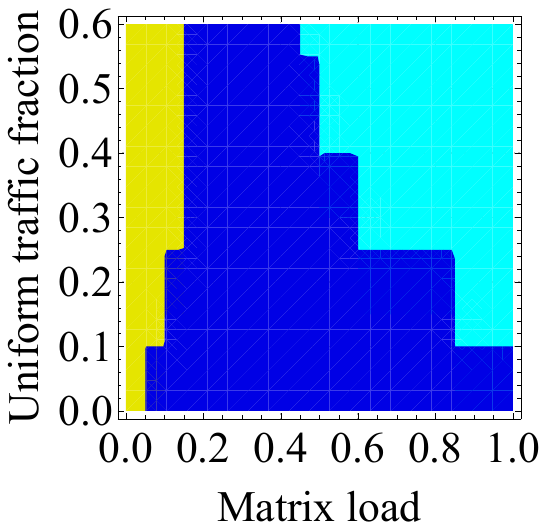} \\
     (a) $\bvnalg$ & (b) $\simalg$ & (c) \gmn 
    \end{tabular}
    \vspace{-11pt}
    \caption{Comparison of the DCT of \csn, \rsn, and \msn. (a) and (b) in the multi-permutation matrix $M(n,m,u,L)$ for $n=64$ and $m=20$, 
    (c) on simulation-based traffic.
    Regions denote the system with the lowest DCT.}  
     \label{fig:2d}
     \vspace{-11pt}
\end{figure}

\subsection{A Case Study}\label{sec:testcase} 
In this section, we consider a simple case study to evaluate the DCT of the three systems.
Let $M(n, m, u, L)$  be an $n \times n$ scaled doubly stochastic demand matrix where the sum of each row and column is equal to $L$. $M(n, m, u, L)$ is built from a uniform matrix $M^u$ with $\card{M^u}=uLn$ and a multi-permutation matrix $M^m$ with 
$\card{M^m}=(1 - u)Ln$. $M^m$ is the union of $m$ scaled permutation matrices, where each row and each column has $m$ entries of $\frac{1-u}{m}L$, and the entries on the diagonal are zero. W.l.o.g.~we assume that $0 < u < 1$ and $1 \leq m \leq (n-2)$. The matrix $M(n, m, u, L)$ is an abstraction of the simulation traffic where a $u$ fraction of the traffic has uniform flow sizes and a $1-u$ fraction of the traffic is mostly sparse with large flows, so decomposable to a few matchings.

In Appendix~\ref{sec:thecasestudy}, we analytically compute the DCT of \csn, \rsn, and two versions of \msn ($\bvnalg$ and $\simalg$) for the case study demand matrices. 
Therefore, for this demand matrix, we can systematically compare when each system has the smallest DCT for $M(n, m, u, L)$. 
\autoref{fig:2d} presents such a 2-dimensional DCT map for the reconfiguration times of the simulation setup, with $n=64$, and $m=20$ matchings. The $Y$ axis of the map is the fraction of uniform traffic $u$, and the $X$ axis is the normalized load $\frac{L}{r}$ where $r$ is the link rate from the simulation. Each point on this map represents a different demand matrix $M(n, m, u, L)$.
We then divide the map into three regions, and each region shows the range of parameters for which each system \csn, \rsn, or \msn has the smallest DCT. 

\autoref{fig:2d} (a) and (b) consider $\bvnalg$ and  $\simalg$ for \msn, respectively. 
There are some interesting observations when comparing the two figures. 
The blue fraction representing the vertical, $\bvnalg$, algorithm is larger in (a), suggesting that, despite that both algorithms send the \emph{same} size $m$ of the permutation set, $\bvnalg$ is more efficient. 
A second observation is to notice the point on the $X$ axis where $\dctda = \dctrot$. Solving for $x$ (by Equations~\eqref{eq:dctmultirot} and~\eqref{eq:dctmultida}) where $x=\frac{L}{r}$ we get $x=\frac{\eta (n-1)\crec}{2-\eta}$ which for our parameters is approximately $x\approx 0.6$. 
Above $x$ at least some matchings will use the \cs scheduler, and below $x$ at least some matchings will be sent to the \rs scheduler.  
\revisiontwo{This value for $x$ gives us the bound between the \csn and \rsn dominated areas. To generate the three different areas in the figure, we also need to find where \msn is dominant (i.e., its DCT is smaller).  Therefore, when, $x\geq 0.6$ (roughly), we need to find where \msn is dominant over \csn; that is where the value given by \autoref{eq:dctmultida} is larger than the DCT value given by \autoref{eq:dctVeralg} for \autoref{fig:2d}~(a) and \autoref{eq:dctSimalg} for \autoref{fig:2d}~(b). When $x< 0.6$ we need to find where \msn is dominant over \rsn that is, where the DCT value given in \autoref{eq:dctmultirot}  is larger than the DCT value given by \autoref{eq:dctVeralg} for figure \autoref{fig:2d}~(a) and \autoref{eq:dctSimalg} for \autoref{fig:2d}~(b).}
Overall, we can see, as expected, that increasing the load will reach a point where \csn is the best system (since the longer reconfiguration time is amortized), and decreasing the load will reach a point where \rsn has a smaller DCT. Looking at the $Y$ axis, we find that with a lower $u$ we observe higher success for \msn. Lower $u$ means a higher fraction is represented by the same $20$ matchings.
This makes the \rs scheduler less effective on them and the \cs scheduler less effective on the uniform traffic, justifying the traffic partition of \msn.

\setlength{\textfloatsep}{10pt}
\begin{algorithm}[t]
    \DontPrintSemicolon
   \caption{The \talg Algorithm}\label{alg:theo:mixnet}
   Input: $M \in \mathbf{R}^{n\times n}$ scaled double stochastic\;
   Initialize output: $P^{(da)}= \emptyset$, $P^{(rot)}= \emptyset$\;
   Decompose $M$ to a set permutation matrices $P(M)$ using maximum matchings\; 
    \ForAll {$P_i \in P(M) $}{
      $M_i = \alpha_i \cdot P_i$\;
      \eIf{ $\dctda(M_i) \le \dctrot(M_i)$ }
      {$P^{(da)}= P^{(da)} \bigcup P_i$}
      {$P^{(rot)}= P^{(rot)} \bigcup P_i$}
    }
    $M^{(da)}=\sum_{P_i \in P^{(da)}} M_i$ and $M^{(rot)}=\sum_{P_i \in P^{(rot)}} M_i$
\end{algorithm}

\nspace
\subsection{The Greedy Algorithm}\label{subsec:divalgo}
\autoref{alg:theo:mixnet} presents \talg (\gmn), a natural and simple algorithm to minimize $\dctmix$ for \msn. 
Given a saturated demand matrix $M$ as input, the algorithm first decomposes $M$ into a set of permutation matrices $P(M)$ and the corresponding scaling factors $\alpha$.
It does so in a standard way, iteratively and greedily, computing a maximum matching in $M$, setting $M_i$ to be a scaled permutation matrix $\alpha_i P_i$ (e.g., by considering $\alpha_i$ as the minimum value in the matching), and removing $M_i$ from $M$. This method is guaranteed to decompose $M$ to a set (not necessarily minimal) of scaled permutation matrices, as proven by the BvN theorem. 
Next, for each $M_i$, we then add $P_i$ either to $P^{(da)}$ if $\dctda(M_i)\leq \dctrot(M_i)$ and otherwise we add it to $P^{(rot)}$.
Since $M_i$ is a scaled permutation matrix, we can easily calculate the \emph{optimal} DCT of each individual matching $M_i$ on both \rsn and \csn and determine the partition of matrices to sub-systems.
For \csn and $\dctda$, using Equation~\eqref{eq:daDCT} with $m=1$ and for \rsn and $\dctrot$ we can use Equation~\eqref{eq:rsnDCT}.
We then define, as described earlier, the traffic that will use the \cs scheduler as  $M^{(da)}=\sum_{P_i \in P^{(da)}} M_i$ and the traffic that will use the \rs scheduler as $M^{(rot)}=\sum_{P_i \in P^{(rot)}} M_i$.

The $\dctmix$ of the algorithm is then defined as in Equation~\eqref{eq:dctmix}, and using Equation~\eqref{eq:daDCT} and~\eqref{eq:rsnDCT} as,
\begin{align}
    \dctmix(M,\gmn) = \frac{\card{M^{(da)}}}{nr}+m\crec + \frac{2}{\eff r} \frac{\card{M^{(rot)}}}{n}.
\end{align}
\nspace

\revisiontwo{We can now bound the DCT of \msn with  \gmn.
\begin{theorem}\label{thm:gmn}
For any saturated demand matrix $M$,
\begin{align*}
     \dctmix(M,\gmn)\leq \min( \dctrot(M, \val), \dctda(M, \bvn))
\end{align*} 
\end{theorem}
}

We can examine the performance of \gmn and \msn on the 
case study demand from \autoref{sec:testcase} (i.e., $M(n,m,u,L)$).
For this case, we can show that \gmn and \msn will match the DCT of the best of the three systems we considered for $M(n,m,u,L)$. For example, \gmn will match the results of each region in \autoref{fig:2d} (a) (see Appendix \ref{sec:thecasestudy}). 

\subsection{Empirical Evaluation of \talg} 
We evaluate \gmn on the same traffic as generated for the simulations in \autoref{sec:evaluation}.
The results, shown in \autoref{fig:2d} (c), are similar \revision{in spirit } to the theoretic results in \autoref{fig:2d} (a) \revision{ \& (b)}.
This shows how our simplified traffic model is able to capture the essence of more realistic traffic.

To understand this figure recall that \talg needs to decompose $M$. We use a simple greedy decomposition algorithm based on BVN in our implementation. Given a demand matrix $M$ we search for the maximal weighted matching and remove it from $M$.
In the general case, we would fully decompose a matrix, however, since each subsequent matching found is getting smaller, we decompose the matrix until there is some fraction of traffic left. 
While we do not explicitly consider this, we assume small traffic can be sent using the \ss sub-topology. Let us denote this fraction as $L_s$. 
In this section, we consider $L_s=0.05$, i.e. $5\%$ of the total traffic is assumed to be sent using the \ss sub-topology.
Since the simulation traffic does not necessarily adhere to the saturation requirement, we cannot assume that $M$ is perfectly saturated. Each matching might also not be saturated. We, therefore, use a cutoff value for the matching. As a simple heuristic, we choose the average value over their original matching as the cutoff. Each value in the found matching will be, at most, the average value after the cutoff.
In Figure \ref{fig:2d} (c), we show results based on five real traffic traces from the simulation as described earlier, in \autoref{subsec:Settings}. 
The specific traces used are with a load $L=40\%$ with five different datamining shares $(40\%, 55\%,70\%,85\%,100\% )$. 
Furthermore, for the sake of consistency between different loads, we normalize all matrices to the global mean matrix value. 
The traffic trace itself is a sequence of ToR-level traffic matrices, each representing a total of $0.05$ seconds of accumulated remaining traffic from the simulation. That is, after each $0.05$ seconds, we gather all pending traffic into an accumulated matrix. 
Each row represents a different datamining load, and each column represents an accumulated trace. That is, at the cell $(0.3,0.7)$, we gather the first $14$ matrices from the $70\%$ trace.
To color the cells we denote cells with less than a fraction of $\frac{1}{12}$ of the total DCT in \rsn pure \csn and vise versa. Otherwise, we assume the optimal system is \msn. 

In Figure \ref{fig:emp:GMNimprove}, we see the DCT results for four algorithms, \talg, $\simalg$, \csn, and \rsn on simulated traces; the traffic is the same as in Figure \ref{fig:2d} (c) for a load of $40\%$.
Note that we could only implement $\simalg$ and not $\bvnalg$ on these traces. Recall that traffic is generated such that a fraction of the load is from uniform traffic, and a portion is sampled from datamining. To implement $\simalg$, we sent the uniform traffic to \rsn and the rest to \csn.  
To implement $\bvnalg$, we would need to know in advance how to decompose the matrix into a set of $m$ 'large' matchings, for \csn, and a set of $n-1-m$ 'small' matchings, for \rsn, which is not normally possible. This is different than our \talg, which tries to find an optimal set $m$ matching to send to \csn

Regarding the results for \rsn, we see that all the results are the same for \rsn. We can explain this by looking at the DCT of \rsn. In \autoref{eq:rsnDCT}, we see that it is a function of the total load and not a function of the structure of the matrix itself.  The \emph{total} traffic of the simulated matrices was normalized to the same expected size, that is $t\cdot L\cdot r\cdot n\cdot k$, where $t$ is the time represented by the matrix in seconds, resulting in a similar outcome. For the other algorithms, the BVN decomposition results in a different DCT due to the different sizes of the decomposition.

To conclude, this figure shows that \talg outperforms all other algorithms, consistently achieving a lower DCT of about $0.88$~s.  
For \csn, the DCT increases with \datamining share, from about 0.94~s to 1.08~s. The DCT of \rsn is 0.982~s for all \datamining shares since all matrices have the same volume.
\begin{figure}
    \centering
    \includegraphics[width=0.9\linewidth]{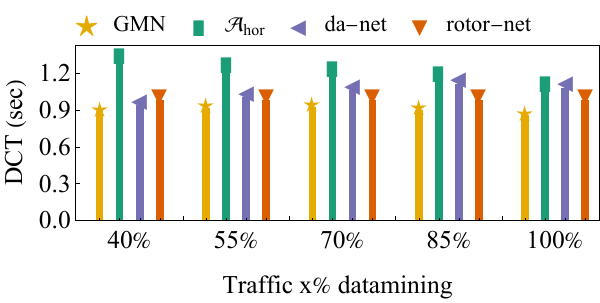}
    \vspace{-10pt}
    \caption{\revision{Empirical results on the 40\% load traces showing the DCT of \talg, $\simalg$, \csn and \rsn.}  
    }
    \label{fig:emp:GMNimprove}
\end{figure}
\nspace

\section{Conclusions}\label{sec:conclusion}
Motivated by the desire to maximize throughput in datacenters whose workloads evolve over time, we presented a self-adjusting datacenter network whose topology can be quickly adapted to dynamically and optimally match the workload requirements.
Our architecture, \system, is enabled by Sirius and relies on a novel decentralized control plane with links and packet scheduling. 

Our work opens several interesting avenues for future research. 
In particular, while we have primarily focused on performance, it will be interesting to study how our architecture can be extended to provide resilience under different failure models. 
\revision{It will also be interesting to further study scalability aspects related to the limited or heterogeneous number of ports or laser frequencies.}

\section*{Acknowledgements}
Research supported by European Research Council (ERC), grant agreement No. 864228 (AdjustNet), Horizon 2020, 2020-2025.

\label{lastpage}
{ \footnotesize \bibliographystyle{unsrt}
\bibliography{bibi.bib}
}

\newpage
\appendix

%%%%%%%%%%%%%%%%%%\appendix%%%%%%%%%%%%%%%%%%%%%%%%%%%%%%%%%%%%%%%%%%%%%%%
\section{Details on \debruijn graph-based topology}\label{sec:appendix:debruijn}

\begin{table}[!htb]
    \centering
    \footnotesize
    \begin{tabular}{l|c|c|l}
        Prefix & Port & Len & IP\\\hline
        $11*$ & $DA$ & $2$ & $10.192.0.0/10$ \\
        $101$ & $\{1, DA\}$ & $3$ & $10.160.0.0/11$ \\
        $011$ & $\{1, DA\}$ & $3$ & $10.96.0.0/11$ \\
        $010$ & $1$ & $3$ & $10.64.0.0/11$ \\
        $001$ & $1$ & $3$ & $10.32.0.0/11$ \\
        $000$ & $0$ & $3$ & $10.0.0.0/11$ \\
        \hline
        $100$ & Local & $0$ & $10.128.0.0/11$
    \end{tabular}
    \caption{Forwarding table at node $100$ in \autoref{fig:trio:debruijn}.}\label{tab:trio:debruijn_fwd_table}
\end{table}
\revisiontwo{
\autoref{tab:trio:debruijn_fwd_table} shows the forwarding table of node $010$ in \autoref{fig:trio:debruijn} (including the \cs port). The table includes the forwarding port IDs, $0, 1$ or $DA$, the path length to destinations using this rule, and the IP representation. Besides the logical addressing in the \debruijn graph, it also illustrates how these addresses can be embedded into IP(v4) addresses and, thereby, how the \debruijn graph-based topology component can be realized with today's packet switching equipment. 
Note also that multiple, equal cost paths are supported.\footnote{More details on embedding the \debruijn address into the IP addresses, ports IDs, and the table update procedure are provided in~\cite{zerwas2023duo}.}}

\section{Additional Material for the Design of \system}

\begin{table}[!htb]
    \centering
    \footnotesize
    \begin{tabular}{c|p{6cm}}
        Variable & Description \\ \hline\hline
        $n$ & Total number of ToRs. \\
        $k$ & Total number of spines \\
        $k_r$ & Total number of rotor spines \\
        $T_i$ & ToR i, $1\leq i\leq n$ \\
        $R_i$ & Rotor i,  $1\leq i\leq \rnum$ \\
        $h=n\cdot k$ & Total number of hosts \\
        $C=r\cdot \delta$ & Slot active capacity (volume to send) \\
        $r$ & Link rate \\
        $\Bar{C} = \frac{C\cdot \rnum}{k}$ & host ToR rotor capacity\\
        $h_{i,j}$ & Host i on ToR j (can be translated to an integer for indexing) \\
        $u_{i,j}$ & Uplink i on ToR j, ($1\leq i \leq \rnum$) \\
        $T_{u_{i,j}}$ & ToR connected on uplink i on ToR j \\
        \hline
        \textbf{Input} & \\ \hline\hline
        $D^l_{i,j}$ & current local demand from host i to host j\\
        $D^n_{i,j}$ & current non-local demand from host i to host j\\
        \hline
        \textbf{Output} & \\ \hline\hline
        $S^{ld}_{i,j}$ & direct local volume to send from host i to j\\
        $S^{n}_{i,j}$ & non-local volume to send from host i to j\\
        $S^{li}_{i,j,k}$ & indirect local volume to send from host i to j with final destination k\\
    \end{tabular}
    \caption{Notation.}
    \label{tab:trio:notation}
\end{table}

\subsection{\system Packet Scheduling (LocalLB)}

This section provides the algorithms underlying \system's packet scheduling.
First, \autoref{tab:trio:notation} summarizes the used notation. In particular, it defines the input (sizes of the local and non-local buffers) and output data structures (allocated volumes from local and non-local to the destination and from local to an intermediate host).
\autoref{alg:scheduling:fairshare} is a sub-routine used by LocalLB. 
The routine is summarized in pseudo-code from prior work (the code base) from \opera~\cite{opera}.
Given a matrix of demand to be sent, it applies a 2-dimensional fair share over the sending and receiving capacities (given as vectors). Therefore, it repeats sweeping the rows (l. 5f) and then the columns (l. 7f) until the output values converged.

\autoref{alg:trio_scheduling} is the full packet scheduling algorithm with the major steps as described in \autoref{ssec:scheduler}. It takes a per ToR perspective (the algorithm runs on ToR $t$); all data structures are rack-local but the algorithm listing partially uses global indices to simplify notation.

In lines 2-6, the algorithm allocates the second hop non-local traffic and local traffic that can be sent directly to the destination. For each buffer type, it iterates over the \rs ports and per port collects the demand information (local or non-local respectively) of all source and destination hosts that are connected by the port.
Then, it allocates the demand using the 2-dimensional fairshare (FS2D, l. 5).
The second part (lines 7- 21) allocates traffic to be sent via an intermediate host. Again, the algorithm iterates over the \rs ports.
Per port, it iterates over the hosts in the destination rack (with the order changing in a round-robin fashion between calls of the packet scheduling).
For each such candidate intermediate host, it searches for the destination with the highest demand per source host (lines 9-13). 
Then, it iterates over these demands in non-decreasing order. If a demand $z_i$ exceeds the remaining capacity in the slot to the intermediate host, the equal-share of the remaining capacity is given to all remaining demands $>0$ (lines 15-18). If the demand does not exceed the remaining capacity to the intermediate host, it is greedily allocated (l. 20f).
Note that the algorithm does not evaluate the capacity of the slot from the intermediate node to the actual destination but relies on \system's offloading (\autoref{alg:trio_offloading}) to compensate for overloads.

\subsection{\system FIB}
To illustrate the content of ToRs' FIBs, we consider a network with eight ToRs. We consider a point in time at which \ss and \cs ports are connected according to the scenario depicted in \autoref{fig:trio:debruijn}. In addition to the ports shown there, each ToR has one \rs port.
The IP addressing scheme follows the description in~\cite{zerwas2023duo}. We use bits 22-24 to encode the \debruijn node addresses. The lower 21 bits can be used to address hosts inside the rack.
The rules match on a combination of exact match (for the slot label) and LPM (for the destination address).
To give an example, \autoref{tab:trio:full_fib} shows the FIB for node 100. 
The upper part (note that the order was chosen arbitrarily and does not necessarily reflect the order in memory/the order of the lookup) shows the forwarding rules for the \ss and \cs part and essentially matches \autoref{tab:trio:debruijn_fwd_table}.
The only addition is an exact match on the slot label to be $0$.
The lower part of the FIB shows the rules for the \rs forwarding. Here, the matching value of the slot label depends on the active slot. The match also considers the destination address to perform forwarding if multiple \rs ports are present.
Packets that do not match any rule are dropped.

\begin{table}[t]
    \centering
    \begin{tabular}{c|c | c | c}
        \# &  Slot Label & Dst. Addr. & Port \\\hline\hline
        1 & 0 & $10.192.0.0/10$ & $DA$ \\
        2 & 0 &$10.160.0.0/11$ & $\{1, DA\}$ \\
        3 & 0 & $10.96.0.0/11$  & $\{1, DA\}$ \\
        4 & 0 &  $10.64.0.0/11$ &  $1$\\
        5 & 0 &$10.32.0.0/11$ & $1$  \\
        6& 0 & $10.0.0.0/11$ &  $0$ \\\hline
        7& 0 &  $10.128.0.0/11$ & Local \\\hline
        8 & 1 & 10.96.0.0/19 & \{\rs\} \\
         9 & 2 & 10.128.0.0/19 & \{\rs\} \\
         10 & 3 & 10.160.0.0/19 & \{\rs\} \\
         11 & 4 & 10.192.0.0/19 & \{\rs\} \\
         12 & 5 & 10.224.0.0/19 & \{\rs\} \\
         13 & 6 & 10.0.0.0/19 & \{\rs\} \\
         14 & 7 & 10.32.0.0/19 & \{\rs\} \\\hline
        15 & * & * & drop
    \end{tabular}
    \caption{Examplary forwarding table for node 100 following the \ss and \cs topologies in \autoref{fig:trio:debruijn} with one additional \rs port.}\label{tab:trio:full_fib}
\end{table}

\begin{algorithm}[!htb]
\caption{\system Offloading (from Rotor to \debruijn) at host $v$}\label{alg:trio_offloading}
        \For{dest $u=1\dots h$, ($v$ and $u$ not in same rack)}{ 
            \eIf{$D^l_{v,u}> d_{\mathrm{off}}$}{
                Offload $D^n_{v,u}$ to backbone topology
            }{
                \If{$D^l_{v,u} + D^n_{v,u} > d_{\mathrm{off}}$}{
                    Offload $D^l_{v,u} + D^n_{v,u} - \frac{C}{k}$ from $D^n_{v,u}$ to backbone
                }
            }
        }
\end{algorithm}

\begin{algorithm}[t]
    \DontPrintSemicolon
    \caption{2-dimensional fair share (from \opera's repository~\cite{opera}).}\label{alg:scheduling:fairshare}
    $I\in \mathbf{R}^{N\times M}, c0\in\mathbf{R}^N, c1\in\mathbf{R}^M$\;
    Initialize output: $O\in \mathbf{R}^{N\times M}$\;
    Count elements $>0$ in $I$ as $n^+$\;
    \While{$n^+>0$ and max iterations not reached}{
        \For{$i=1\dots N$}{
            $t[i]$ = 1-dimensional fair share on $I[i]$ with capacity $c0[i]-\sum_{j=1\dots M}O[i,j]$\;
        }
        \For{$j=1\dots M$}{
            $t2$ = 1-dimensional fair share on $j$-th column of $t$ with capacity $c1[j]-\sum_{i=1\dots N}O[i, j]$\;
        }
        Update $t$ with $t2$\;
        Update $I$, $O$ with $t$\;
        Zero out rows $I[i]$ where $c0[i]=0$\;
        Zero out columns $I[j]$ where $c1[j]=0$\;
        Count elements $>0$ in $I$ as $n^+$\;
    }
\end{algorithm}

\begin{algorithm}[t]
    \DontPrintSemicolon
    \caption{\system Packet-Scheduling (LocalLB) for ToR $t$}\label{alg:trio_scheduling}
     $C^{TX}_i\gets \Bar{C}$, $C^{RX}_{i,b}\gets \frac{C}{k}, \forall i=1\dots k, \forall b=1\dots k_r$\;
    \tcp{Allocate Non-local (n) on the 2nd hop and local (l)}
     \For{x $\in \{n, l\}$ }{
         \For{$b=1\dots \rnum$}{
             $W_{i,j}\gets D^x_{h_{i,t},h_{j,T_{b,t}}}\ \forall i=1\dots k, j=1\dots k$\;
             $S^x \gets FS2D(W, C^{TX}, C^{RX})$\;
             Update $C^{RX}, C^{TX}, D^x$\;
         }
     }
     \tcp{Allocate new indirect traffic on the remaining capacity}
     \For{$b=1\dots \rnum$}{
        \tcp{Iterate over all hosts in ToR connected via $b$}
         \For{$\forall a\in \left\{h_{j, T_{b,t}}, j=1\dots k\right\}$, order RR over slots}{
            \tcp{Find the largest demand per source host capped by remaining sending capacity}
             $V_{i,m}\gets D^{l}_{h_{i,t},h_{m}} \forall i=1\dots k,m=1\dots nk$\;
             $v_{i,m}=v_{i,m}-\frac{C}{k}$ and zero all $<0$\;
             $y_i\gets\arg\max_m v_{i,m}, \forall i=1\dots k$\; 
             $z_i\gets \min\{v_{i,y_i}, C^{TX}_i \}$\;
             $n_z\gets$ number of $z_i >0$\;
             \For{$\forall i$ where $z_i>0$ in non-decreasing order of $z_i$}{
                 \eIf{$z_i > C^{RX}_{a,b}$}{
                     Allocate $\frac{C^{RX}_{a,b}}{n_z}$ for all $i$ with $z_i > 0$\;
                     Update $C^{TX}, C^{RX}, D^l$\;
                     Break\;
                 }{
                     Allocate $z_i$ for $i$ and update $C^{TX}, C^{RX}, D^l$\;
                     Decrement $n_z$\;
                 }
             }
         }
     }
\end{algorithm}
\newpage

\begin{table*}[t]
    \centering
    \footnotesize
    \begin{tabular}{|m{.5cm}|c|c|c|c||c|c|c||c|c|c|}
        \hline
        & & \multicolumn{3}{c||}{$10$\,G} & \multicolumn{3}{c||}{$50$\,G} & \multicolumn{3}{c|}{$100$\,G} \\
        Type & Parameter & Opera  & \system & \duo  & Opera  & \system & \duo & Opera  & \system & \duo\\ \hline\hline
        \multirow{3}{*}{\rotatebox[origin=c]{90}{\scriptsize{\rs}}} & $\rrec$   &    $10\mu s + 1.7\mu s$  & $100ns + 1.7\mu s$ & NA & $10\mu s + 0.74\mu s$  & $100ns + 0.74\mu s$ & NA & $10\mu s + 0.62\mu s$  & $100ns + 0.62\mu s$ & NA\\
        & $\delta$  &  $33.46\cdot\rrec$ &  $54.56\cdot\rrec$ & NA &  $13.65\cdot\rrec$ &  $14.65\cdot\rrec$ & NA &  $10.93\cdot\rrec$ &  $12.88\cdot\rrec$ & NA  \\
        & Duty cycle & $97\%$ & $98\%$ & NA & $93\%$ & $93\%$ & NA & $92\%$ & $92\%$ & NA \\\hline
        \multirow{3}{*}{\rotatebox[origin=c]{90}{\shortstack{\scriptsize{\textsc{Demand}} \\ \scriptsize{\textsc{-aware}}}}} & $\crec$   &   NA   & $1ms = 10k \cdot\rrec$ & $1ms$ &   NA   & $100\mu s = 2k \cdot\rrec$ & $100\mu s$ &   NA   & $100\mu s = 1k \cdot\rrec$ & $100\mu s$  \\
        & $\delta_d$  & NA & $49 \crec$ & $49 \crec$ & NA & $49 \crec$ & $49 \crec$ & NA & $49 \crec$ & $49 \crec$ \\
        & Duty cycle & NA & $98\%$ & $98\%$ & NA & $98\%$ & $98\%$ & NA & $98\%$ & $98\%$ \\
        \hline
    \end{tabular}
    \caption{Reconfiguration parameters.}
    \label{tab:eval:parameters}
    \label{tab:eval:parameters_100G}
    \label{tab:eval:parameters_50G}
\end{table*}

\section{Analysis of Case Study}\label{sec:thecasestudy}
We first consider the DCTs of \rsn and \csn. 
For \rsn, we can use Equation~\eqref{eq:rsnDCT}, so 
\begin{align}
    \dctrot(M(n, m, u, L))=2\frac{\card{M(n, m, u, L)} }{\eff n r} =\frac{2L}{\eff r}. \label{eq:dctmultirot}
\end{align}

For \csn, we can decompose the matrix into the $m$ original permutation matrices and additionally a set of $(n-1-m)$ permutation matrices for the rest of the uniform matrix. In total, we have $n-1$ matchings. Using Equation~\eqref{eq:daDCT}, we get the following result,

\begin{align}
    \dctda(M(n, m, u, L))=\frac{L}{r}+ (n-1)\crec. \label{eq:dctmultida}
\end{align}
In both cases, we note that the DCT is a function of the 
the sum of a row $L$.

Next, we study \msn. For this, we need first to find a partition of $M(n, m, u, L)$ into $M^{(da)}$ and $M^{(rot)}$.
We consider two potential algorithms: $\simalg$ and $\bvnalg$.

The $\simalg$ algorithm (horizontal) is motivated by the simulation setup and sets $M^{(da)} = M^m$ and $M^{(rot)}= M^u$ where all uniform traffic is sent with the \rs link scheduler, and all the permutation matrices are sent with the \cs scheduler. The $\dctmix$ in this case is

\begin{align}\label{eq:dctSimalg}
\dctmix(M(n, m, u, L), \simalg) &=\dctda(M^{m})+\dctrot(M^u) \nonumber\\
&= u\frac{L}{r}+ m \crec+ (1-u)\frac{2L}{\eff r}.
 \end{align}

For the $\bvnalg$ algorithm (vertical), we again decompose $M(n, m, u, L)$  into a set of $m$ matchings sent with the \cs scheduler, but here the total traffic in the $m$ matchings also includes traffic from the uniform matrix; the remaining matrix is sent with the \rs scheduler. 
The DCT is determined by the average row size of $M^{(da)}$ and $M^{(rot)}$.
In $M^{(da)}$, a cell form one of the $m$ matchings, has $ L\frac{1-u}{m}$ bits from $M^m$ and $L\frac{u}{n-1}$ bits from $M^u$.
For the remaining entries, each of the $(n-1-m)$ non-zero entries has $L\frac{u}{n-1}$ bits. The $\dctmix$ is then
\begin{align}\label{eq:dctVeralg}
&\dctmix(M(n, m, u, L), \bvnalg) = \notag \\
&\frac{L m}{r}(\frac{ 1-u}{ m}+\frac{u}{n-1})+ m\crec+ \frac{2Lu}{\eff r} (\frac{n-1-m}{n-1}).   
\end{align}  

\subsection{Formal results}
Formally for $M = M(n,m,u,L)$, we can show the following:
\begin{claim}\label{clm:gmn}
Let $M$, a demand matrix be $M = M(n,m,u,L)$, than the following holds for the algorithm $\gmn$.
\begin{align}
    &\dctmix(M,\gmn) = \notag \\
    &=\min\{\dctmix(M, \bvnalg), \dctda(M,\bvn), \dctrot(M,\val)\}
\end{align}   
\end{claim}
\begin{proof}
    After decomposing $M$ using \bvn, we first observe that there are only two sizes of permutations matrices in the decomposition. The first, has a size of $\frac{uL}{n-1}$ per row. These permutations were made directly from the matrix $M^u$. And an other set of permutations made from the overlap between  $M^u$ and $M^m$, these will have a row size of $\frac{(1-u)L}{m(n-1)}+\frac{uL}{n-1}$.
    Since \gmn tests each permutation on an individual basis, it has exactly three options. Send all permutations to \csn, send all permutations to \rsn, or send the smaller of the two types to \rsn and the larger to \csn. Since the second set with the size of  $\frac{(1-u)L}{m(n-1)}+\frac{uL}{n-1}$ is always larger $\gmn$ will be emulating $\bvnalg$. Lastly, note that $\gmn$ will never send the smaller permutations to \csn and the larger to \rsn. This is true for any matrix, as it's a results of the DCT of both networks being linear. 
\end{proof}

The following is a simple observation of Claim \ref{clm:gmn}.

\begin{observation}\label{obs:PermOnMsn}
    Let $P\in P_\pi$ be a permutation matrix. \msn has  an equal DCT to either \csn or \rsn for $P$,
    \begin{align}
     \dctmix(P,\gmn) = \min( \dctrot(P, \val), \dctda(P, \bvn))
\end{align} 
\end{observation}
Note that $P$ is a special case in $M(n,m,u,L)$. The observation stems from the fact that $\gmn$ only schedules the entire matrix $P$ on either \msn or \rsn.

Next, we prove Theorem \ref{thm:gmn} which holds for any saturated demand matrix.

\begin{theorem}
For any saturated demand matrix $M$,
    \begin{align}
     \dctmix(M,\gmn)\leq \min( \dctrot(M, \val), \dctda(M, \bvn))
\end{align} 
That is, the DCT of \msn will be equal or less than the DCT of either \rsn or \csn. 
\end{theorem}

\begin{proof}
To see why this is the case, let us denote the set of matchings sent on \csn as $M^{(da)} =\{M_i : \dctda(M_i)\leq \dctrot(M_i)\}$. And let $M^{(rot)}= \{M_i : \dctrot(M_i)\leq \dctda(M_i)\}$ be the set of matchings sent on \rsn. 
Since for an individual permutation matrix \msn has lower or equal DCT to either \csn or \rsn as from \autoref{obs:PermOnMsn}, and since both functions describing the DCT either system (in Equation~\eqref{eq:daDCT} and Equation~\eqref{eq:rsnDCT}) are linear in the total load of either $M^{(da)}$ or  $M^{(rot)}$, the results follows.
\end{proof}
\label{lastpageappendix}

\end{document}